\documentclass[12pt]{article}
\usepackage{}
\usepackage{} 
\usepackage[sectionbib]{natbib}
\usepackage{array,epsfig,fancyheadings,rotating}
\usepackage[]{hyperref}  
\usepackage{arydshln}
\usepackage[marginal]{footmisc} 
\usepackage{sectsty, secdot}
\sectionfont{\fontsize{12}{14pt plus.8pt minus .6pt}\selectfont}
\renewcommand{\theequation}{\thesection\arabic{equation}}
\subsectionfont{\fontsize{12}{14pt plus.8pt minus .6pt}\selectfont}

\usepackage{amsmath}
\usepackage{amssymb}
\usepackage{amsfonts}
\usepackage{multirow}
\usepackage{amsthm}
\usepackage{indentfirst}
\usepackage{color}
\usepackage{CJK}

\usepackage{cancel}
\usepackage{bm}
\usepackage{amsmath,amsfonts,amssymb} 
\usepackage{bbm}
\usepackage{color}
\usepackage{graphicx}
\usepackage{epstopdf}
\usepackage{epsfig}

\usepackage{indentfirst}
\usepackage{array}
\usepackage{booktabs}
\usepackage{float}
\usepackage{multirow}
\usepackage{diagbox}

\newtheorem{theorem}{Theorem}

\theoremstyle{definition}

\newtheorem{algorithm}{Algorithm}
\theoremstyle{plain}
\newtheorem*{Lem*}{Lemma}

\newtheorem{Lem}{Lemma}[section]

\numberwithin{equation}{section}

\newcommand{\ignore}[1]{}

\newcommand{\beginsupplement}{%
	\setcounter{table}{0}
	\renewcommand{\thetable}{S\arabic{table}}%
	\setcounter{figure}{0}
	\renewcommand{\thefigure}{S\arabic{figure}}%
}

\newtheorem{assump}{Assumption}

\usepackage{algorithm}
\usepackage{algorithmic}

\def\E{{\mathbb{E}}}
\def\X{{\bm X}}

\def\I{{\mathbb I}}
\def\P{{\rm P}}

\usepackage{xr}
\usepackage{hyperref}
\begin{document}
	\begin{CJK*}{GBK}{song}
		
		
		\renewcommand{\baselinestretch}{2}
		
%
%
		\renewcommand{\thefootnote}{}
		
		
		\fontsize{12}{14pt plus.8pt minus .6pt}\selectfont \vspace{0.8pc}
		\centerline{\large\bf LEVERAGE CLASSIFIER: ANOTHER LOOK AT }
		\vspace{2pt}
		\centerline{\large\bf SUPPORT VECTOR MACHINE}
		\vspace{.4cm}
		
		\centerline{Yixin Han\textsuperscript{1}, Jun Yu\textsuperscript{2}, Nan Zhang\textsuperscript{3}, Cheng Meng\textsuperscript{4}, Ping Ma\textsuperscript{5}, Wenxuan Zhong\textsuperscript{5}, and Changliang Zou\textsuperscript{1}} \vspace{.4cm}
		
		\centerline{\it \textsuperscript{1}School of Statistics and Data Science, LPMC $\&$ KLMDASR, Nankai University, Tianjin, P.R. China}
		
		\centerline{\it\textsuperscript{2}School of Mathematics and Statistics, Beijing Institute of Technology, Beijing, P.R.China}
		
		\centerline{\it \textsuperscript{3}School of Data Science, Fudan University, Shanghai, P.R.China}
		
		\centerline{\it  \textsuperscript{4}Institute of Statistics and Big Data, Renmin University, Beijing, P.R.China}
		
		\centerline{\it \textsuperscript{5}Department of Statistics, University of Georgia, Athens, GA, USA}

		\vspace{.55cm} \fontsize{9}{11.5pt plus.8pt minus .6pt}\selectfont

		
		\begin{quotation}
			\noindent {\it Abstract:}
			Support vector machine (SVM) is a popular classifier known for accuracy, flexibility, and robustness. However, its intensive computation has hindered its application to large-scale datasets. In this paper, we propose a new optimal leverage classifier based on linear SVM under a nonseparable setting. Our classifier aims to select an informative subset of the training sample to reduce data size, enabling efficient computation while maintaining high accuracy. We take a novel view of SVM under the general subsampling framework and rigorously investigate the statistical properties. We propose a two-step subsampling procedure consisting of a pilot estimation of the optimal subsampling probabilities and a subsampling step to construct the classifier. We develop a new Bahadur representation of the SVM coefficients and derive unconditional asymptotic distribution and optimal subsampling probabilities without giving the full sample. Numerical results demonstrate that our classifiers outperform the existing methods in terms of estimation, computation, and prediction. 
			
			\vspace{9pt}
			\noindent {\it Keywords and phrases:}
			{Classification; Large-scale dataset; Martingale; Optimal subsampling; Support vector machine.}
		\end{quotation}\par
		
		\footnotetext{\noindent { Corresponding author: pingma@uga.edu  (Ping Ma)}}
		
		\def\thefigure{\arabic{figure}}
		\def\thetable{\arabic{table}}
		
		\renewcommand{\theequation}{\thesection.\arabic{equation}}
		\fontsize{12}{14pt plus.8pt minus .6pt}\selectfont
		
		\setcounter{section}{0} 
		\setcounter{equation}{0} 
		
		\lhead[\footnotesize\thepage\fancyplain{}\leftmark]{}\rhead[]{\fancyplain{}\rightmark\footnotesize\thepage}
		
		\section{Introduction}\label{sec:intro}
		
		Consider the binary classification problem for a training sample of size $N$, $\mathcal{D}_N=\{(\X_j, Y_j)\}_{j=1}^N$,
		where  $\X_j\in\mathbb{R}^p$ denotes covariates (a.k.a.features), $Y_j=\left\{1,-1\right\}$ represents class labels. The central task is to build a classifier that predicts the label based on the observed covariates. Numerous literature is available on binary classification procedures, including nearest neighbor classifiers, discriminant analysis, logistic regression, tree-based methods, support vector machine, and ensemble learning. See, for example, \citet{hastie2010elements,fan2020statistical} for a comprehensive review.

		Support vector machine (SVM) is a theoretically motivated classifier and has gained significant popularity in various applications \citep{boser1992training,cortes1995support,vapnik2013nature}. As a margin-based approach, SVM aims to find the maximum-margin hyperplane in either the original or extended kernel feature space. According to the elegant geometric interpretation, only a subset of the training dataset called the \emph{support vectors}, needs to be considered for evaluating the separating hyperplane. This property is attractive compared to likelihood-based classifiers, such as logistic regression, which depend on all training data to determine the discriminative boundary. Moreover, logistic regression is typically fitted under the assumption that the response follows a binomial distribution, whereas SVM does not require any distributional assumption and thus leads to more robust performance \citep{steinwart2008support}. 
		
		Despite the advantages mentioned above, constructing an SVM classifier is computationally intensive as it typically involves solving quadratic programming optimization problems. In general, the computational cost of SVM is $O(N^2 N_s )$ \citep{kaufman1998solving}, where $N_s$ represents the number of support vectors.  In practice,  $N_s$ usually increases linearly with the sample size $N$ of the training data. As a result, the number of support vectors  
		significantly affects the training time and the evaluation of the decision boundary.  Various methods have been proposed to mitigate the computational complexity of training SVM classifiers. For example, specialized algorithms for solving quadratic programming have been suggested, including the sequential minimal optimization \citep{platt1998fast} and various decomposition methods used in the LibLinear software library \citep{hsieh2008dual}. Other fast computation methods based on low-rank approximation \citep{williams2001using}, gradient descent \citep{bordes2005fast,shalev2011pegasos,wang2012breaking}, core set \citep{tsang2005core}, and nearest neighbor \citep{camelo2015nearest} have also been developed.  However, it is worth noting that most of these methods still incur a computational cost of at least ${O}(N^2)$ or lack optimal statistical guarantees. Therefore, when the sample size of the training data is huge, both time complexity and statistical guarantees become prohibitively demanding.
		
		Observing that the discriminative boundary of the SVM depends on only a subset of the training data, we take another look at the SVM from the perspective of data reduction. A crucial insight from the SVM is that a relatively small subset of the training data is sufficient to build up an effective classifier. Inspired by leverage score sampling methods developed for least-squares regression \citep{drineas2011faster,ma2015statistical} and low-rank matrix approximation \citep{mahoney2009cur}, our strategy is to construct an importance sampling distribution for all the training data points, which effectively reduces the data size before constructing the classifier. The nonuniform subsampling strategy we employ is straightforward to design and implement.  As long as the reduced dataset remains informative or representative, the corresponding estimator can provide a satisfactory approximation to the estimator based on the full sample. For example, the statistical leveraging framework \citep{drineas2012fast,ma2015statistical,ma2022asymptotic,li2020modern} has achieved great success in large-scale ordinary least squares regression.  
		More recently, optimal subsampling procedures have been also established for various statistical models, including logistic regression \citep{wang2018optimal}, generalized linear models \citep{ai2018optimal,yu2022optimal}, quantile regression \citep{wang2021optimal}, nonparametric regression \citep{ma2015efficient,meng2020more,meng2021lowcon}, and designed for testing problems \citep{ren2022large, Han2023model}. However, none of the existing can be directly applied to SVM due to its distinguishing geometric feature.  Consequently, our goal is to develop a leverage classifier that is computationally efficient for large datasets and theoretically provable as the SVM.
		
		In this paper, we introduce a novel binary classifier based on linear SVM in a nonseparable setting. To construct the optimal classifier, we propose a two-step optimal subsampling algorithm that involves a pilot study to estimate the optimal subsampling probabilities and a subsampling step. Our subsampling procedure significantly reduces the computational costs without scarfing too much estimation efficiency.  With a novel view of the SVM under the general subsampling framework, we rigorously investigate the statistical properties of the proposed classifier. Specifically, we derive the asymptotic distribution and the optimal subsampling probabilities. Our contributions can be summarized as follows:
		\begin{itemize}
			\item [(1)] Double randomnesses are addressed: one arising from the training data and the other from the subsampling procedure. Our approach yields an unconditional asymptotic result regardless of the full sample and thus allows for random subsampling probabilities.
			
			\item [(2)] We utilize the martingale technique as observations in the selected samples are no longer independent. Our theoretical framework builds upon the Bahadur representation of the linear SVM estimator, which is nonstandard in the context of the general subsampling strategy. 
			
			\item [(3)] The nonuniform subsampling probabilities are computed by minimizing specific criteria derived from the asymptotic variance, leading to  optimality within the experimental design theory. Numerical results also demonstrate that our leverage classifier is computationally fast, and the identified separating hyperplane is close to that obtained using the full sample SVM.
		\end{itemize}

		The remainder of this paper is organized as follows. Section \ref{background} reviews the linear SVM for nonseparable binary classification and motivates the leverage classifier framework. Section \ref{method} investigates the theoretical properties of leverage classifiers and develops efficient algorithms for constructing optimal leverage classifiers. Simulation studies and a real-world example are presented in Sections~\ref{Simulation}--\ref{CASP}. Section \ref{conclusion} concludes the paper with some potential improvements. All theoretical proofs and additional numerical results are provided in the Supplementary Material. The implementing codes for this work are available in \url{https://github.com/yuxiaohaihyx0517/Leverage-Classifier}.
		
		\section{Support vector machine and leverage classifier}\label{background}

		\subsection*{2.1~~~Support vector machine}
		
		Binary linear classification problem aims to find the best separating hyperplane of the form
		$f(\X,\bm\beta)=\beta_0+\X^{\top}{\bm\beta}_1,$ with intercept $\beta_0$ and slope vector ${\bm\beta}_1$. Write ${\bm\beta}=\left(\beta_0,{\bm\beta}_1^\top\right)^\top\in\mathbb{R}^{p+1}$ and $\widetilde{{\X}}=\left(1,{\X}^\top\right)^\top\in\mathbb{R}^{p+1}$ as the augmented parameter and data vectors, and then $f({\X},\bm\beta)=\widetilde{{\X}}^\top\bm\beta$. When the training data are not linearly separable, the linear SVM solves the following optimization problem
		\begin{align}\label{LinearSVMmargin}
			\widehat{\bm\beta} = &\mathop{\arg\min}_{\bm\beta\in\mathbb{R}^{p+1}}\left\{\frac{1}{N}\sum\limits_{j=1}^N\left[1-Y_jf({\X}_j,\bm\beta)\right]_{+}+\frac{{\lambda_{\textrm{FULL}}}}{2}\|\bm\beta_1\|^2\right\},
		\end{align}
		where $\left[u\right]_{+}=\max(u,0)$ is the hinge loss function, $\|\cdot\|$ denotes the Euclidean norm of a vector, and the tuning parameter $\lambda_{\textrm{FULL}}>0$ controls the amount of regularization on model complexity. 
		
		From the theoretical perspective, \cite{koo2008bahadur} investigated the asymptotic behavior of the coefficient of the linear SVM. Denote the population version of the loss function in \eqref{LinearSVMmargin} without penalty by
		$L(\bm\beta)= {\E}\left[1-Yf({{\X}},\bm\beta)\right]_+,$
		and its minimizer $\bm\beta^{\dagger}=\arg\min_{\bm\beta}L(\bm\beta)$.
		Define
		\begin{align*}
			{\bm S}(\bm\beta)=-{\E}\left\{{\I}\left(Yf({{\X}},\bm\beta)\leq 1\right)Y\widetilde{{\X}}\right\},\quad
			{\bf H}(\bm\beta)={\E}\left\{\psi \left(1-Yf({{\X}},\bm\beta)\right)\widetilde{{\X}}\widetilde{{\X}}^\top\right\},
		\end{align*}
		where ${\I}(\cdot)$ is the indicator function and $\psi(\cdot)$ is the Dirac delta function.
		Provided that ${\bm S}(\bm\beta)$ and ${\bf H}(\bm\beta)$ are well-defined \citep{koo2008bahadur}, they are interpreted as the gradient and Hessian matrix of $L(\bm\beta)$. Subsequently, under regularity conditions, $\widehat{\bm\beta}$ satisfies
		\begin{align}\label{fullasynormal}
			\sqrt{N}(\widehat{\bm\beta}-{\bm\beta}^{\dagger}) {\rightarrow} \mathcal{N}\left({\bf 0},{\bf H}(\bm\beta^{\dagger})^{-1}{\E}\{{\I}(Yf({{\X}},\bm\beta^{\dagger})\leq 1)\widetilde{{\X}}{\widetilde{{\X}}}^\top\}{\bf H}(\bm\beta^{\dagger})^{-1}\right).
		\end{align}

		From an optimization perspective, the representer theorem \citep{kimeldorf1971some,scholkopf2001generalized} states that the solution to the quadratic programming in \eqref{LinearSVMmargin} admits a finite-dimensional expression of basis functions. In general, solving a quadratic programming optimization problem has a computational cost of ${O}(N^3)$ \citep{mehrotra1992implementation,chang2011psvm}, which becomes prohibitively expensive when the training data size $N$ is large. However, in the case of the linear SVM,  a significant fraction of the basis coefficients can be zero. The training data associated with the nonzero basis coefficients are called support vectors, which play a crucial role in determining the discriminative boundary. As a result, the computational cost is significantly reduced as the number of support vectors is much smaller than the training sample size, making it more feasible for large-scale datasets.
		
		
		\subsection*{2.2~~~Leverage classifier~~}\label{Generalsubsampling}
		
		Inspired by the appealing property of support vectors, we revisit the SVM and develop a new classifier called leverage classifier. Our strategy first selects an informative subset of the training data with some nonuniform subsampling probabilities and then constructs the linear SVM classifier based on the reduced dataset. The leverage classifier integrates leverage score sampling with the margin-based classifier, and its advantage is to approximate the discriminative boundary well with significantly reduced computational cost. In our subsampling framework, we employ the subsampling with replacement strategy to ensure theoretical convenience. The detailed procedure of the leverage classifier is described in Algorithm \ref{GeneralAlgorithm}.
		
		\begin{algorithm}[htbp]
			\caption{Leverage classifier.}
			\begin{algorithmic}\label{GeneralAlgorithm}
				\STATE \textbf{Step 1} Assign subsampling probabilities ${\bm\pi}=\{\bm\pi_j\}_{j=1}^N$ to all training samples in $\mathcal{D}_N$;
				
				\STATE \textbf{Step 2}
				Draw a subset of size $n\ll N$ from $\mathcal{D}_N$ according to ${\bm\pi}$ via subsampling
				with replacement. Denote the subsample by $\mathcal{S}_n=\{({\X}_i^*,Y_i^*)\}_{i=1}^n$ and the corresponding subsampling
				probabilities by $\bm\pi^*=\{\bm\pi^*_i\}_{i=1}^n$; \medskip
				
				\STATE \textbf{Step 3}  Use $\mathcal{S}_n$ to train the linear SVM by minimizing the penalized weighted hinge loss with a properly tuned parameter $\lambda$
				\begin{align*}
					\widetilde{{\bm\beta}} = \mathop{\arg\min}_{\bm\beta\in\mathbb{R}^{p+1}}\left\{\frac{1}{n}\sum\limits_{i=1}^n\frac{\left[1-Y_i^*f({{\X}}_i^*,\bm\beta)\right]_{+}}{N\bm\pi_i^*}+\frac{{\lambda}}{2}\|\bm\beta_1\|^2\right\}.
				\end{align*}
				\STATE \textbf{Step 4} The separating hyperplane is $f({\X},\widetilde{\bm\beta})=\widetilde{{\X}}^\top\widetilde{\bm\beta}$.
			\end{algorithmic}
		\end{algorithm}
		
		The performance of the leverage classifier relies on the subsampling probability $\bm\pi$, the subsample size $n$, and the tuning parameter $\lambda$. First, the reduced dataset $\mathcal S_n$ is obtained according to $\bm\pi$. A simple choice, $\pi_j=N^{-1}$, leads to uniform subsampling. 
		Although this strategy is useful for exploratory data analysis, it often fails to extract important information by ignoring the distinctive characteristics of statistical models. Recent studies on logistic regression \citep{wang2018optimal} and quantile regression \citep{wang2021optimal} have highlighted the importance of designing nonuniform subsampling strategies. Our subsequent analysis reveals that the leverage classifier, with carefully designed $\bm\pi$, can attain a certain level of optimality in terms of experimental design. Second, \citet{kaufman1998solving} pointed out that the number of support vectors typically  increases linearly with the training sample size.  As a result, the leverage classifier with $\mathcal S_n$ of size $n$ offers a more efficient computational approach compared to the SVM utilizing the full sample size $N$.  Lastly, training the leverage classifier involves tuning parameter selection, which differs from the aforementioned literature. We employ the Generalized Approximate Cross-Validation method (GACV). Specifically, minimize objective function $N^{-1}\sum\nolimits_{k=1}^N[1-Y_kf_\lambda^{[-k]}(\X_k,\bm\beta)]_+$, where $f_\lambda^{[-k]}(\X_k,\bm\beta)$ is the SVM solution with $k$-th data point removed. This objective function stems from the penalized likelihood estimates in SVM and serves as a generalization of the generalized cross-validation. 	GACV does not need to train and test every possible hyperparameter combination and thus is a computationally efficient method. See \citet{wahba2003optimal} for its optimal properties and implementation details.

		Before proceeding with theoretical analysis, we provide a toy example to illustrate the intuition of the leverage classifier. Please refer to Section~\ref{Simulation} for the implementation details. In Figure~\ref{toy}, 
		the right panel showcases the best separating hyperplane determined solely by the support vectors associated with the full sample SVM. The left panel displays the leverage classifier with A-optimality (explained in Section~\ref{method}), which tends to select data points close to the full sample support vectors, resulting in a reduced dataset that is informative in identifying the discriminative boundary.  
		In contrast, the middle panel demonstrates the uniform subsampling strategy, which overlooks the characteristics of the full sample support vectors. As a result, the selected subsample is less informative.  Unless the subsample size $n$ is relatively large, the uniform subsampling strategy will be inferior to a carefully designed nonuniform subsampling strategy used by the leverage classifier.
		
		\begin{figure}[htbp]
			\centering
			\includegraphics[width=1\textwidth]{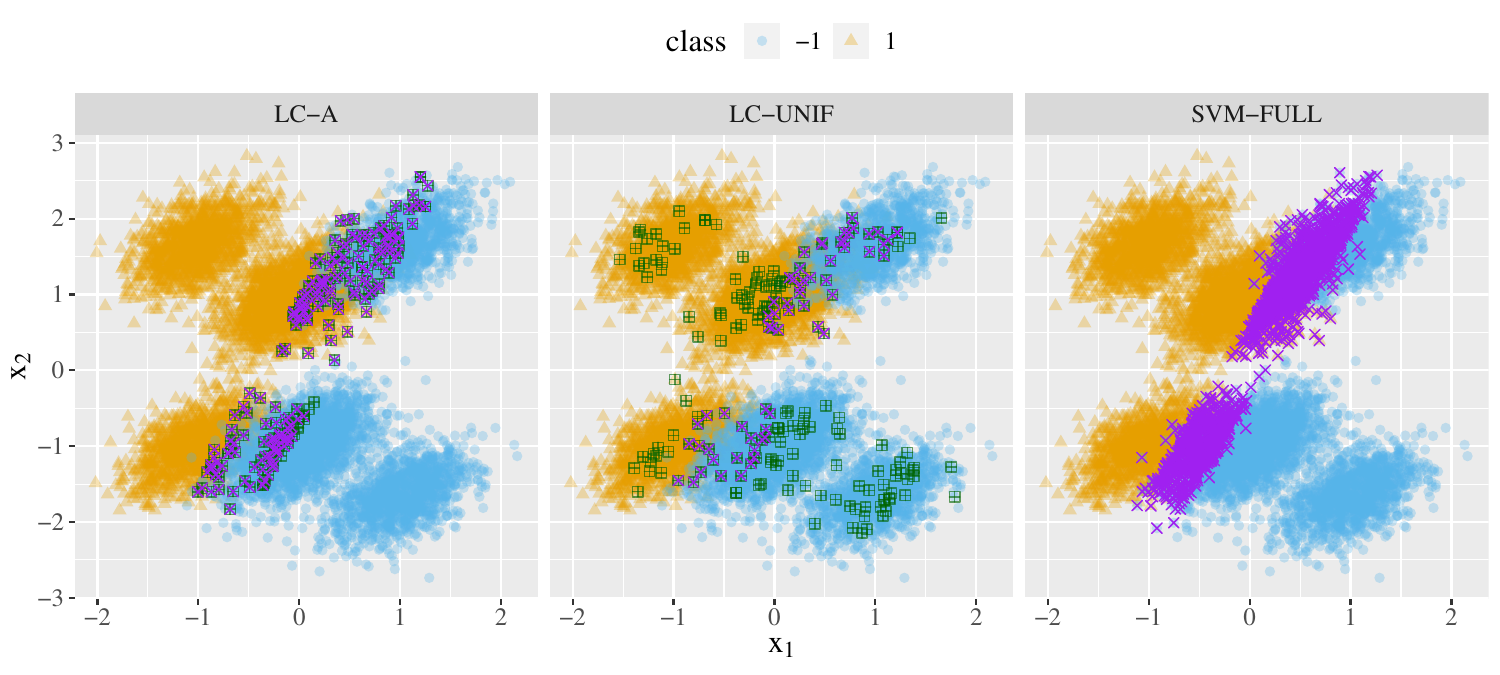}\par
			\vspace{-0.2cm}
			\caption{Toy example for linear classification. Classifiers are the proposed optimal leverage classifier with A-optimality (LC-A), the leverage classifier with uniform subsampling (LC-UNIF), and the full sample linear SVM (SVM-FULL). The green {\tiny $\boxplus$}'s denote the selected subsamples, and the purple {$\times$}'s denote the support vectors.}\label{toy}
		\end{figure}

		\section{Theoretical properties and optimal leverage classifier}\label{method}
		
		In this section, we establish theoretical properties and provide an efficient algorithm for the proposed leverage classifiers under the subsampling framework. 
		
		\subsection*{3.1~~~Asymptotic normality}\label{TheoryNormality}

		\begin{assump}\label{density}
			The conditional densities of ${{\X}}$ given class $Y=1$ and $Y=-1$ with respect to the Lebesgue measure are continuous and have finite fourth moments.
		\end{assump}
		\begin{assump}\label{betadiff}
			The covariates for the two classes have different mean values in at least one dimension.
		\end{assump}
		\begin{assump}\label{derivativeHS}
			The nonzero minimizer $\bm\beta^{\dagger}$ of $L(\bm\beta)$ is unique and satisfies that ${S}({\bm\beta}^{\dagger})=0$.
			${\bf H}(\bm\beta)$ is positive-defined around $\bm\beta^{\dagger}$ in a compact set $\mathcal{B}$ with a nonzero radius.
		\end{assump}
		\begin{assump}\label{approxexpansion}
			The subsampling probabilities satisfy that
			\begin{align*}
				\frac{1}{N^3}\sum\limits_{j=1}^N\mathbb{E}\left(\frac{1}{\pi_j^2}\right)=O(1).
			\end{align*}
		\end{assump}
		
		Assumptions \ref{density}--\ref{derivativeHS} are commonly imposed to establish the asymptotic normality of the linear SVM, and they typically hold under the regularity conditions outlined in \citet{koo2008bahadur}. Assumption~\ref{approxexpansion} allows for random subsampling probabilities since the full dataset is not fixed. Furthermore, Assumption~\ref{approxexpansion} restricts $\bm\pi$ from being extremely small,  preventing any training sample from dominating the weighted penalized hinge loss function in Step~3 of Algorithm \ref{GeneralAlgorithm}.  
		When we condition on the full dataset, Assumption \ref{approxexpansion} is in the similar spirit of the commonly used subsampling schemes, for example, \citet{ai2018optimal,wang2018optimal}.
		
		\begin{theorem}[The Bahadur representation]\label{Bahadur-type representation}
			Suppose Assumptions \ref{density}--\ref{approxexpansion} hold. For $\lambda = o(n^{-1/2})$, we have
			\begin{equation}\label{eq:Bahadur-rep}
				\sqrt{n}(\widetilde{\bm\beta}- {\bm\beta^{\dagger}}) = -\frac{1}{\sqrt{n}}{ \bf H}({\bm\beta}^{\dagger})^{-1}\sum\limits_{i=1}^n\frac{1}{N\pi_i^*}\xi_i^*Y_i^*\widetilde{{{\X}}}_i^*+o_{P}\left(1\right),
			\end{equation}
			where $\xi_i^*={\I}\left(Y_i^*f({{\X}}_i^*,{\bm\beta}^{\dagger})\leq 1\right)$ and $\widetilde{{\X}}_i^*=\left(1,{{\X}_i^{*\top}}\right)^\top $, $i=1,\ldots,n$.
		\end{theorem}
		
		Theorem~\ref{Bahadur-type representation} presents a Bahadur representation of $\widetilde{\bm\beta}$ for the leverage classifier under the subsampling framework, which is the building block for establishing the asymptotic normality. As discussed in \citep{koo2008bahadur}, the condition $\lambda = o(n^{-1/2})$ is an appropriate order for nonseparable SVM, and additional simulation results confirm the rationality of this condition. The use of subsampling with replacement and the integration of the subsampling probability makes Theorem~\ref{Bahadur-type representation} a nontrivial extension of \citet{koo2008bahadur}, which only considered SVMs learned from independent and identically distributed data. The Bahadur representation reveals how the subsampling strategy and margins of the optimal separating hyperplane determine the statistical behavior of the estimator.

		Next, we establish the unconditional asymptotic normality of $\widetilde{\bm\beta}$ based on the Bahadur representation. To this end, we define  ${\bm T}={n}^{-1}\sum_{i=1}^n(N\pi_i^*)^{-1}\xi_i^*Y_i^*\widetilde{{\X}}_i^*$ as a term on the right hand side of \eqref{eq:Bahadur-rep}.  As Algorithm~\ref{GeneralAlgorithm} conducts subsampling with replacement, the data in the reduced dataset $\mathcal S_n$ are no longer independent unless conditioned on the full training sample. Hence, we treat the subsampling procedure as a stochastic process and employ the martingale technique to study the asymptotic property of $\bm T$.  Let {${{\X}}_1^N=\left({{\X}}_1,\ldots,{{\X}}_N\right)$} and ${{Y}}_1^N=\left({{Y}}_1,\ldots,{{Y}}_N\right)$. Step~2 in Algorithm~\ref{GeneralAlgorithm} can be viewed as a $n$-step sequential sampling procedure: in the $i$-th step, we select one data point with replacement from the full training sample and denote it by $({\X}_i^*, Y_i^*)$.  Let $\sigma(*_i)$ be the $\sigma$-algebra \citep{durrett2019probability} generated by the $i$-th sampling step, which is closed under complement, countable unions, and countable intersections.  Accordingly, we thus define a filtration as $\mathcal{F}_{N,0}=\sigma\left({{\X}}_1^N,{Y}_1^N\right)$ and $\mathcal{F}_{N,i}=\sigma\left({{\X}}_1^N,{Y}_1^N\right)\vee\sigma\left(*_1\right)\vee\cdots\vee\sigma\left(*_i\right)$ for $i=1,\dots,n$. This filtration $\mathcal{F}_{N, i}$ be explained as the smallest $\sigma$-algebra containing all the information after the $i$-th sampling step. Based on this filtration, we define ${\bm M}=\sum_{i=1}^n {\bm M}_i$, where
		\begin{align*}
			{\bm M}_i=\frac{1}{nN\pi_i^*}\xi_i^*Y_i^*\widetilde{{\X}}_i^*-\frac{1}{nN}\sum\limits_{j=1}^{N}\xi_jY_j\widetilde{{\X}}_j.
		\end{align*}
		We can express $\bm T=\bm M+ \bm Q$ with ${\bm Q}={N}^{-1}\sum_{j=1}^N\xi_jY_j\widetilde{{\X}}_j$, where above decomposition allows for decoupling the variabilities from the sampling procedure and the full dataset, which are measured by $\bm M$ and $\bm Q$, respectively. In the Supplementary Material, we demonstrate that $\left\{{ \bm M}_i,i=1,\ldots,n\right\}$ forms a martingale difference sequence adapted to filtration $\left\{\mathcal{F}_{n,i},i=1,\ldots,n\right\}$. Using the martingale central limit theorem \citep{ohlsson1989asymptotic}, we establish the unconditional asymptotic normality of $\widetilde{\bm\beta}$.  
		
		\begin{theorem}[Asymptotic normality]\label{martingalenormality}
			Suppose Assumptions \ref{density}--\ref{approxexpansion} hold. Then the variance of $\bm T$, denoted by ${\bf V}_T$, can be written as
			\begin{align*}
				{\bf V}_{T} = \frac{1}{nN^2}\sum\limits_{j=1}^N\E_{Y\mid \X}\left(\frac{1}{\pi_j}{\I}\left(Y_jf({\X}_j,\bm\beta^{\dagger})\leq 1\right)\widetilde{{\X}}_j\widetilde{{\X}}_j^\top\right)+{ {\bf C}},
			\end{align*}
			where $\bf C$ is a constant matrix that does not depend on $\bm\pi$. As $N\to\infty$, $n\to\infty$, 
			we have
			\begin{align*}
				{\bf V}^{-1/2}(\widetilde{\bm\beta}-{\bm\beta}^{\dagger}) {\rightarrow} \mathcal{N}({\bf 0}, {\bf I}_{p+1}),
			\end{align*}
			in distribution, where ${\bf V}={ \bf H}({\bm\beta}^{\dagger})^{-1}{\bf V}_{ T}{\bf H}({\bm\beta}^{\dagger})^{-1}$ and ${\bf I}_{p+1}$ is the identity matrix of dimension $p+1$.
		\end{theorem}
		
		Theorem \ref{martingalenormality} typically allows for random $\bm\pi$ since the subsampling probabilities may depend on the response. When $\bm\pi$ is prespecified or does not depend on $Y$, the variance can be further simplified to ${\bf V}_{T} = ({nN^2})^{-1}\sum_{j=1}^N{\pi_j}^{-1}{\P}\left(Y_jf({\X}_j,\bm\beta^{\dagger})\leq 1\right)\widetilde{{\X}}_j\widetilde{{\X}}_j^\top+{ {\bf C}}$. In this case, the subsampling procedure affects all the data points, making it impossible to identify the support vectors without any information about $Y$. Assumptions \ref{approxexpansion} and the moment condition in Assumption \ref{density} are utilized to verify the martingale version of the Lindeberg-Feller conditions. In the proof of Theorem~\ref{martingalenormality}, we observe that the first term in ${\bf V}_T$ is derived from the variance of $\bm M$, while the second term $\bf C$ comes from $\bm Q$ and some terms in the variance of $\bm M$ that are independent of $\bm \pi$. In particular, when $n/N\to 0$, the variability from the full dataset is insignificant. This evokes us to determine optimal subsampling probabilities by minimizing certain criteria based on the first term of ${\bf V}_T$. 
		
		
		\subsection*{3.2~~~Optimal leverage classifier}\label{OSanalysis}
		
		The leverage classifier enables fast computation by using a reduced dataset $\mathcal S_n$. Take the uniform subsampling strategy with $\pi_j^{\text{UNIF}}=N^{-1}$, $j=1,\dots, N$ as an example. Assumption \ref{approxexpansion} is satisfied, and thus the corresponding leverage classifier admits the asymptotic properties described in Theorems~\ref{Bahadur-type representation} and \ref{martingalenormality}. However, the uniform subsampling procedure does not account for any statistical model assumption and may fail to capture the most informative sample points leading to unsatisfactory estimates; see Figure~\ref{toy} for illustration. 
		
		We next explore how to determine the subsampling probabilities $\bm\pi=\{\pi_j\}_{j=1}^N$, by which the leverage classifier attains certain statistical optimality based on the asymptotic properties. A key observation is that in Theorem~\ref{martingalenormality} the asymptotic variance matrix ${\bf V}$ is a function of the subsampling probabilities. It motivates us to derive nonuniform subsampling probabilities by minimizing some criterion associated with ${\bf V}$. To this end, we borrow the concepts from the design of experiments and consider A- and L-optimality criteria \citep{atkinson2007optimum}. Note that we expect the subsampling probabilities to satisfy Assumption~\ref{approxexpansion} although it is not required in the following theorem. We will provide a fix for this issue shortly afterward.
		
		\begin{theorem}\label{Samplingprob}
			When minimizing the traces of 	${\bf V}$ and ${\bf V}_{T}$,  two sets of optimal subsampling probabilities based on A- and L-optimality are
			\begin{align}\label{probAL}
				\begin{split}
					\pi_j^{\mathrm{A}}&=\frac{ { {\I}}\left(Y_jf({{\X}}_j,{\bm\beta}^{\dagger})\leq 1\right)\|{ \bf H}({\bm\beta}^{\dagger})^{-1}\widetilde{{{\X}}}_j\|}{ \sum\limits_{k=1}^N{ {\I}}\left(Y_kf({{\X}}_k,{\bm\beta}^{\dagger})\leq 1\right)\|{ \bf H}({\bm\beta}^{\dagger})^{-1}\widetilde{{{\X}}}_k\|},\\
					\pi_j^{\mathrm{L}}&=\frac{ { {\I}}\left(Y_jf({{\X}}_j,{\bm\beta}^{\dagger})\leq 1\right)\|\widetilde{{{\X}}}_j\|}{ \sum\limits_{k=1}^N{ {\I}}\left(Y_kf({{\X}}_k,{\bm\beta}^{\dagger})\leq 1\right)\|\widetilde{{{\X}}}_k\|},
				\end{split}
			\end{align}
			where $j=1,\dots,N$. Correspondingly, the traces of ${\bf V}$ and ${\bf V}_{T}$ attain their minima.
		\end{theorem}
		
		Theorem~\ref{Samplingprob} takes an optimization approach to deriving the subsampling probabilities by minimizing the traces of ${\bf V}$ and ${\bf V}_{T}$ in Theorem~\ref{martingalenormality}, respectively. The indicator functions ${\I}(Y_jf({\X}_j,\bm\beta^{\dagger})\leq 1)$ in \eqref{probAL} are related to the definition of support vectors, implying that the leverage classifier inherits the virtue of SVM. Moreover, this result differs substantially from the literature, e.g., \citet{wang2018optimal}, which focuses on fixed subsampling probabilities by conditioning on the full dataset. The random response variable $Y_j$ enters into the expressions \eqref{probAL} via ${\I}\left(Y_jf({{\X}}_j,{\bm\beta}^{\dagger})\leq 1\right)$. Given the full dataset, our result will degenerate to fix subsampling probabilities.

		Two issues arise when applying the subsampling probabilities \eqref{probAL} in practice. First, several population quantities, including the true parameter $\bm\beta^{\dagger}$, the Hessian matrix ${ \bf H}(\bm\beta^{\dagger})$, and the indicator function $\I\left(Y_jf({{\X}}_j,\bm\beta^{\dagger})\leq 1\right)$, need to be estimated. Second, the appearance of indicator functions in \eqref{probAL} may lead to a breakdown of Assumption~\ref{approxexpansion}. To address them, we propose to conduct {\it a pilot} study and substitute the unknown population quantities with their corresponding pilot estimates; and apply {\it an additional thresholding} to the indicator functions. 
		
		Specifically, for the pilot study, we select a pilot sample $\mathcal{S}_0=\{(\X_{i0}^*, Y_{i0}^*)\}_{i=1}^{n_0}$ with some proper probabilities $\bm\pi_{0}^*=\{\pi_{i0}^*\}_{i=1}^{n_0}$ from $\mathcal{D}_N$, for instance, using a simple uniform subsampling procedure. We can then replace the true value of $\bm\beta^{\dagger}$ with the pilot estimator $\widetilde{\bm\beta}^0$. Moreover, the Hessian matrix can be estimated using a nonparametric method, as suggested by \citet{koo2008bahadur},	\begin{align}\label{estimateH}
			\widetilde{{ \bf H}}(\widetilde{\bm\beta}^0)&=\frac{1}{n_0}\sum\limits_{i=1}^{n_0}\frac{1}{N\pi_{i0}^*}K_h\left(1-Y_{i0}^*f({{\X}}_{i0}^*,\widetilde{\bm\beta}^0)\right)\widetilde{{{\X}}}_{i0}\widetilde{{{\X}}}_{i0}^{\top},
		\end{align}
		where $K_h(t)=K(t/h)/h$ with bandwidth $h\to 0$ and the kernel function $K(\cdot)$ satisfying $K(t)\geq 0$ and $\int_{-\infty}^\infty K(t)\,{\rm{d}}t=1$.
		The indicator $\I(Y_jf({{\X}}_j,\bm\beta^{\dagger})\leq 1)$ can be replaced by ${\I}(Y_jf({{\X}}_j,\widetilde{\bm\beta}^0)\leq 1)$. 
		For the additional thresholding for the indicator functions, we work under the level $\delta_N>0$ such that
		\begin{align}\label{piIH1}
			\begin{split}
				\widehat{\pi}_{j}^{\text{A}}&=\frac{\max\left\{ {\I}\left(Y_jf({{\X}}_j,\widetilde{\bm\beta}^0)\leq 1\right)\|\widetilde{ \bf H}(\widetilde{\bm\beta}^0)^{-1}\widetilde{{{\X}}}_j\|,\delta_N\right\}}{ \sum\limits_{k=1}^N\max\left\{{\I}\left(Y_kf({{\X}}_k,\widetilde{\bm\beta}^0)\leq 1\right)\|\widetilde{ \bf H}(\widetilde{\bm\beta}^0)^{-1}\widetilde{{{\X}}}_k\|,\delta_N\right\}},\\
				\widehat{\pi}_{j}^{\text{L}}&=\frac{ \max\left\{{\I}\left(Y_jf({{\X}}_j,\widetilde{\bm\beta}^0)\leq 1\right)\|\widetilde{{{\X}}}_j\|,\delta_N\right\}}{ \sum\limits_{k=1}^N\max\left\{{\I}\left(Y_kf({{\X}}_k,\widetilde{\bm\beta}^0)\leq 1\right)\|\widetilde{{{\X}}}_k\|,\delta_N\right\}},
			\end{split}
		\end{align}
		where $\widetilde{\bm\beta}^0$ is the pilot estimate of $\bm\beta^{\dagger}$, and $\delta_N$ is a user-specified constant. If we choose $\delta_N\propto N^{-1}$, the estimated subsampling probabilities \eqref{piIH1} strike a balance between \eqref{probAL} and uniform subsampling probabilities. A simple calculation can verify that the estimated subsampling probabilities \eqref{piIH1} meet Assumption \ref{approxexpansion}, and the asymptotic results follow with $\bm\pi^*$ replaced by $\widehat{\bm\pi}^{\text{A}}$ and $\widehat{\bm\pi}^{\text{L}}$. The two-step optimal leverage classifier is summarized in Algorithm~\ref{OptimalAlgorithm}. 	
		\begin{algorithm}[htbp]
			\caption{Optimal leverage classifier.}
			\begin{algorithmic}\label{OptimalAlgorithm}
				\STATE \textbf{Step 1} Select $n_0$ pilot training samples $\mathcal{S}_0=\{(\X_{i0}^*,Y_{i0}^*)\}_{i=1}^{n_0}$ with subsampling probabilities $\bm\pi_0^*$ from $\mathcal D_N$. Obtain the pilot estimates $\widetilde{\bm\beta}^{0}$ and $\widetilde{\bf H}(\widetilde{\bm\beta}^{0})$;
				
				\STATE \textbf{Step 2} Calculate the optimal subsampling probabilities $\widehat{\bm\pi}^{\text{A}}$ and $\widehat{\bm\pi}^{\text{L}}$ as in \eqref{piIH1};
				
				\STATE \textbf{Step 3} Sample $n$ training samples as $\mathcal{S}_n=\{({\X}_i^*,Y_i^*)\}_{i=1}^n$ with $\widehat{\bm\pi}^{\text{A}}$ and $\widehat{\bm\pi}^{\text{L}}$ from $\mathcal{D}_N$;
				
				\STATE \textbf{Step 4} Implement Algorithm~\ref{GeneralAlgorithm} with $\mathcal{S}_0\cup\mathcal{S}_n$ and a proper tuning parameter $\lambda$ to obtain $\widetilde{{\bm\beta}}$ and the separating hyperplane $f({\X},\widetilde{\bm\beta})=\widetilde{{\X}}^\top\widetilde{\bm\beta}$.
			\end{algorithmic}
		\end{algorithm}

		The choice of the pilot sample size $n_0$ involves a trade-off between estimation efficiency and computational complexity. A larger $n_0$ makes a more precise pilot estimate of $\bm\beta^{\dagger}$ and the Hessian matrix estimation which is estimated by the nonparametric method. However, the computational complexity of the pilot study  should be negligible compared to those in Steps 3 and 4. Hence, we prefer a relatively small $n_0$; Please refer to the Supplementary Martial for a practical recommendation for $n_0$ with empirical evidence. Moreover, it is worth noting that the combination of  $\mathcal{S}_0$ and $\mathcal{S}_n$ in Step~4 maximizes the utilization of selected samples for hyperplane estimation. To obtain the final subsampling estimate in Step 4, we tune $\lambda$ using the weighted version of GACV, which minimizes ${n}^{-1}\sum\nolimits_{k=1}^n{(N\bm\pi_k^*)}^{-1}{\left[1-Y_k^*f_{\lambda}^{[-k]}({{\X}}_k^*,{{\bm\beta}})\right]_{+}}$.

		The overall computational complexity of the optimal leverage classifier comprises three components. First, the cost of the pilot estimates is ${O}(n_0^3)$. Second, calculating the subsampling probabilities $\widehat{\bm\pi}^{\text{A}}$ and $\widehat{\bm\pi}^{\text{L}}$ requires ${O}(N(p+1)^2)$ and ${O}(N(p+1))$, respectively. Third, constructing the the separating hyperplane $\widetilde{{\bm\beta}}$ with $\mathcal{S}_0\cup\mathcal{S}_n$ takes ${O}\left((n+n_0)^3\right)$. In sum, the computational complexities of optimal leverage classifiers with $\widehat{\bm\pi}^{\text{A}}$ and $\widehat{\bm\pi}^{\text{L}}$ are ${O}\left(n_0^3 + N(p+1)^2 + (n+n_0)^3\right)$ and ${O}\left(n_0^3 + N(p+1) + (n+n_0)^3\right)$, respectively. For extremely large $N$, the computational complexity is reduced to ${O}\left(N(p+1)^2\right)$ and ${O}\left(N(p+1)\right)$, which is linear in $N$. Compared with $O(N^3)$ for the full sample SVM, the optimal leverage classifier achieves fast computation with provable optimality.
		
		We conclude with a discussion on the Fisher consistency of the leverage classifier. Fisher consistency is a desirable property of the loss function used by classifiers, that is, the population minimizer of the loss function leads to the Bayes optimal rule of classification \citep{lin2004note}. \citet{lin2002statistical} has shown that the hinge loss function used by the SVM satisfies Fisher consistency for classification. Under the framework of the leverage classifier as in Algorithms~\ref{GeneralAlgorithm}, it is clear that $\E\left([1-Y^*f(\X^*,\bm\beta)]_{+}\right)=\E\left\{\E\left([1-Y_i^*f(\X_i^*,\bm\beta)]_{+}|\mathcal{D}_N\right)\right\}=\E\left([1-Yf(\X,\bm\beta)]_{+}\right)$, which implies that the leverage classifier inherits the Fisher consistency from SVM.

		\section{Simulation Studies}\label{Simulation}

		In this section, we conduct extensive simulated experiments  to demonstrate the numerical performance of our optimal leverage classifiers from the perspectives of estimation, prediction, and computation. 
		
		\subsection*{4.1~~~Settings}
		We generate a set of data points with covariate dimension $p=8$ and randomly split them into two halves as training and testing datasets. The training dataset $\mathcal D_N$ is of size $N=10^5$. 
		The testing dataset is used to evaluate the prediction accuracy. We uniformly sample $n_0=500$ pilot samples for the pilot study.
		All simulation results are based on 500 replications.
		Table \ref{Tab:BandwidthSelector} in our Supplementary Material discusses the selection of bandwidth for Hessian matrix estimation and shows that the effect of different bandwidths can be ignorable. Therefore, we employ Silverman's rule of thumb \citep{silverman1986density} to determine the appropriate bandwidth. We set the thresholding constant in \eqref{piIH1} as $\delta_N=0.01N^{-1}$. For a scalar $c$, write ${\bm c}_p = (c,\dots,c)$ be the $p$-dimensional row vector of $c$'s. Four scenarios are considered:
		\begin{itemize}
			\item [(I)] im-Uniform. The covariate ${\X}$ is independent and identically distributed from the uniform distribution.
			The $l$-coordinate of $\X$ is
			$U[0,1]$ given $Y=1$ and is $U[0.3,1.3]$ given $Y=-1$, $l=1,\dots,p$. The proportions of data points for two classes are 80\% and 20\%. This is an imbalanced case.
			
			\item [(II)] normMIX.
			The covariate ${\X}$ follows a mixture of three multivariate normal distributions with the same covariance matrix but different means. Let ${\X} \sim 0.5\mathcal{N}\left({\bm\mu}_{11},{\bm\Sigma}\right)+0.25\mathcal{N}\left({\bm\mu}_{12},{\bm\Sigma}\right)+0.25\mathcal{N}\left({\bm\mu}_{13},{\bm\Sigma}\right)$ given $Y=1$, ${\X} \sim 0.5\mathcal{N}\left({\bm\mu}_{-11},{\bm\Sigma}\right)+0.25\mathcal{N}\left({\bm\mu}_{-12},{\bm\Sigma}\right)+0.25\mathcal{N}\left({\bm\mu}_{-13},{\bm\Sigma}\right)$ given $Y=-1$,
			where ${\bm\mu}_{11}= ({\bf 0}_{p/2}, {\bf 3}_{p/2})^\top$, ${\bm\mu}_{12}=(-{\bf 3}_{p/2}, {\bf 5}_{p/2})^\top$, ${\bm\mu}_{13}=-{\bf 3}_p^\top$, ${\bm\mu}_{-11}=({\bf 0}_{p/2},-{\bf 3}_{p/2})^{\top}$, ${\bm\mu}_{-12}=({\bf 3}_{p/2},-{\bf 5}_{p/2})^{\top}$, and ${\bm\mu}_{-13}=({\bf 3}_{p/2},{\bf 5}_{p/2})^{\top}$. The proportions of two classes are equal to 50\%.

			\item [(III)]  \text{T3}. The covariate ${\X}$ follows a multivariate $t(3)$ distribution with different means. Let ${\X}  \sim t_3\left({\bm\mu}_1,{{\bf I}}_p\right)/10$ given $Y=1$ and ${\X} \sim t_3\left({\bm\mu}_{-1},{{\bf I}}_p\right)/10$ given $Y=-1$ , where ${\bm\mu}_1 = {\bf 0.75}_p$, ${\bm\mu}_{-1}=-{\bf 0.75}_p$. The proportions of two classes are equal to 50\%.
			
			\item [(IV)]  T3MIX. The covariate ${\X}$ follows a mixture of two multivariate $t(3)$ distributions with different means. Let  ${\X} \sim 0.3t_3\left({\bm\mu}_{11},{{\bf I}}_p\right)+0.7t_3\left({\bm\mu}_{12},{{{\bf I}}}_p\right)$ given $Y=1$ and  ${\X} \sim 0.4t_3\left({\bm\mu}_{-11},{{\bf I}}_p\right)+0.6t_3\left({\bm\mu}_{-12},{{{\bf I}}_p}\right)$ given $Y=-1$, where ${\bm\mu}_{11}= {\bf 2}_p^{\top}$, ${\bm\mu}_{12}=-{\bf 3}_p^{\top}$, ${\bm\mu}_{-11}=-{\bf 1}_p^{\top}$, ${\bm\mu}_{-12}={\bf 8}_p^{\top}$. The proportions of two classes are equal to 50\%.
		\end{itemize}

		\begin{figure}[htbp]
			\centering
			\setlength{\abovecaptionskip}{0cm}
			\setlength{\belowcaptionskip}{-0.cm}
			\includegraphics[width=1\textwidth]{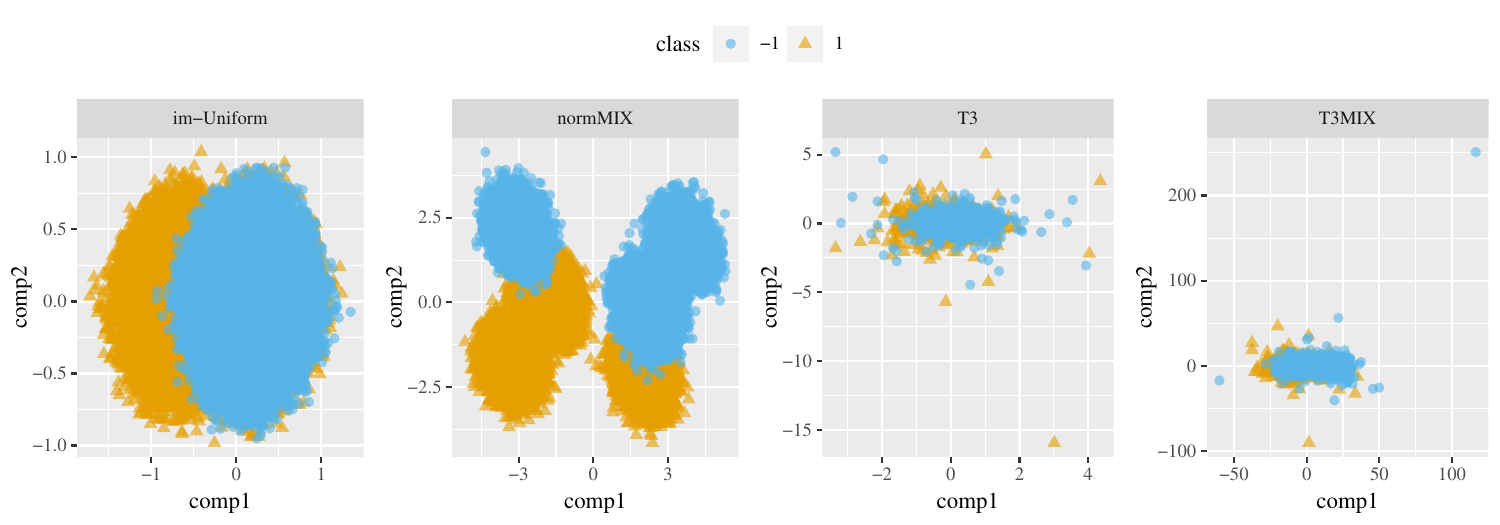}\par
			\vspace{-0.2cm}
			\caption{Full dataset visualization with principal component analysis under Scenarios I--IV.}\label{visualization4}
		\end{figure}
		
		We first project the full datasets of Scenarios I--IV into their first two principal components in Figure \ref{visualization4} to make an intuitive visualization. Besides the optimal leverage classifiers, we also consider Algorithm \ref{GeneralAlgorithm} with $n+n_0$ subsamples uniformly sampled from the training set and the full sample SVM, termed as LC-UNIF and SVM-FULL, respectively.

		\subsection*{4.2~~~Results}
		To assess the estimation performance in approximating the full sample SVM, we calculate the mean squared error of $\widetilde{\bm\beta}$ on training set from $B=500$ replications as $\text{MSE}(\widetilde{\bm\beta}) =B^{-1}\sum_{b=1}^B\|\widetilde{\bm\beta}^{(b)}-\widehat{\bm\beta}\|^2$,
		where $\widetilde{\bm\beta}^{(b)}$ is the estimator obtained  from the $b$-th replication, and $\widehat{\bm\beta}$ is the estimator of the full sample SVM.

		Figure~\ref{plane-estimation-MSE} investigates the effect of subsample size on the estimation performance. Across all simulation scenarios, the optimal leverage classifiers outperform those with uniform subsampling, which aligns with our theoretical analysis in Theorem \ref{Samplingprob}.  The leverage classifier with A-optimal subsampling probabilities performs slightly better than that with L-optimality since A-optimality captures more sample information via the Hessian matrix. Moreover, the proposed methods outperform the leverage classifier with uniform subsampling under Scenario III (T3) and Scenario IV (T3MIX), where the heavy-tail distribution violates the moment assumption in Theorem~\ref{martingalenormality}. As our method is designed to identify points close to the classification hyperplane, it is expected to be robust to outliers. Under the imbalanced case in Scenario I, the optimal leverage classifiers also perform well. Additional simulations in Supplementary Material demonstrate that our method is not sensitive to the pilot sample size $n_0$. Then, we practically recommend the ratio $n_0/(n+n_0)$ to be around $(0.2,0.4)$.
		
		\begin{figure}[htbp]
			\centering
			\setlength{\abovecaptionskip}{0cm}
			\setlength{\belowcaptionskip}{-0.cm}
			\includegraphics[width=1\textwidth]{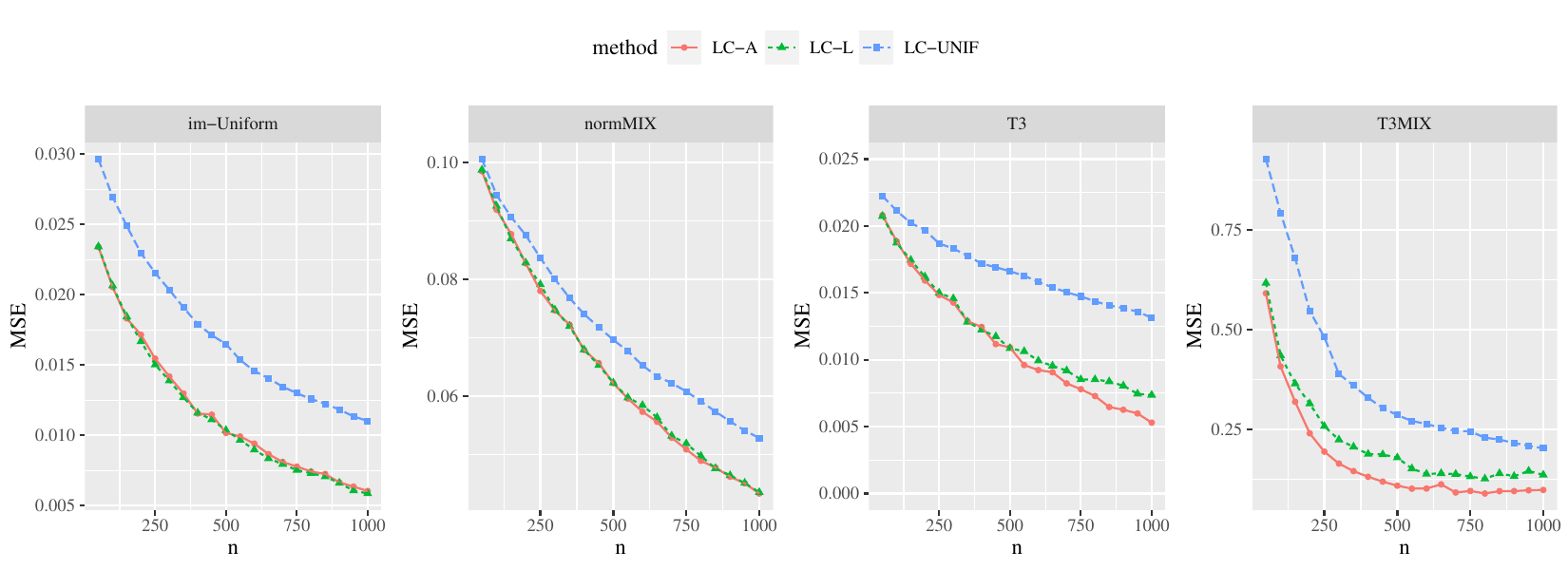}\par
			\vspace{-0.2cm}
			\caption{Comparison of MSE  for approximating the full sample SVM estimator $\widehat{\bm\beta}$ against different subsample sizes under Scenarios I--IV.}\label{plane-estimation-MSE}
		\end{figure}

		\begin{figure}[htbp]
			\centering
			\setlength{\abovecaptionskip}{0cm}
			\setlength{\belowcaptionskip}{-0.cm}
			\includegraphics[width=1\textwidth]{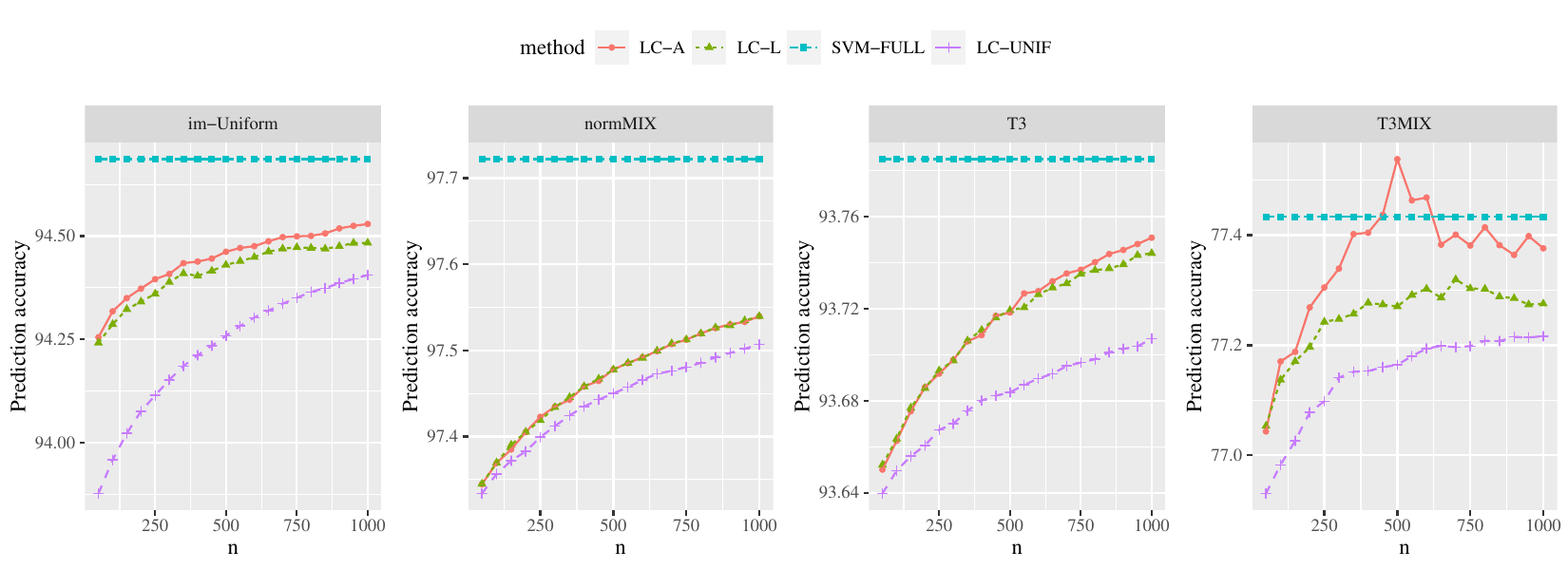}\par
			\vspace{-0.2cm}
			\caption{Comparison of prediction accuracy (\%) against different subsample sizes under Scenarios I--IV.}\label{subsample-accuracy}
		\end{figure}
		
		Figure \ref{subsample-accuracy} indicates that all methods approach the performance of the full sample SVM as $n$ increases.  Remarkably, our optimal leverage classifier sometimes outperforms the full sample SVM in terms of prediction accuracy, as observed in Scenario IV. When $n$ is relatively small, our optimal leverage classifiers exhibit higher prediction accuracy than uniform subsampling, even in scenarios with heavy-tail covariate distribution and imbalanced classes.  In addition, as pointed out by a reviewer, constructing classifiers using the support vectors from the pilot sample degenerates to the special case with $n=0$ of LC-UNIF, which is typically challenging to outperform our optimal classifiers due to the larger subsample size $n$ and optimal subsampling probability $\bm\pi$ utilized in our approach.
		
		\begin{figure}[htbp]
			\centering
			\setlength{\abovecaptionskip}{0cm}
			\setlength{\belowcaptionskip}{-0.cm}
			\includegraphics[width=1\textwidth]{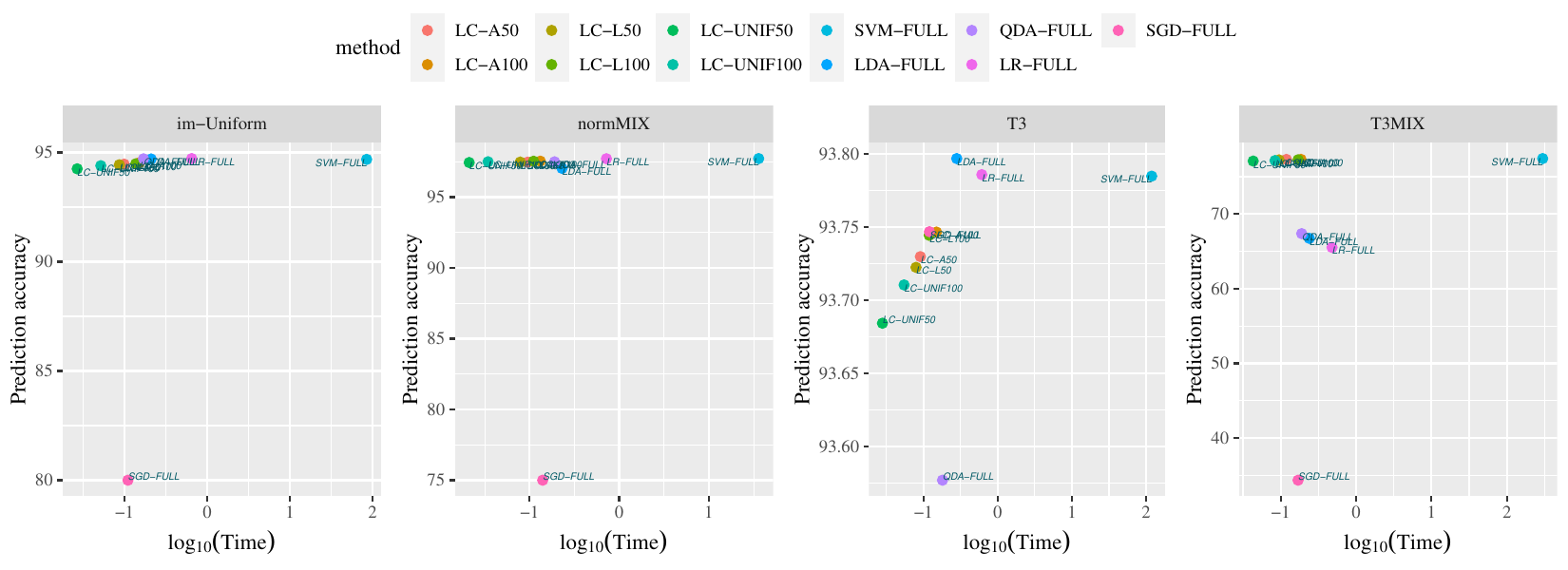}\par
			\vspace{-0.2cm}
			\caption{Comparison of prediction accuracy (\%) and training time for several classifiers against different subsample sizes under Scenarios I--IV. The logarithm is taken on time for a better presentation of the figures.}\label{classifier4}
		\end{figure}
		
		Next, we compare the leverage classifiers with several benchmark classifiers, including logistic regression (LR), linear discriminant analysis (LDA), quadratic discriminant analysis (QDA), and fast stochastic gradient descent (SGD), 
		in terms of training time and prediction accuracy. All four competitors are trained based on the full dataset. 
		Figure \ref{classifier4} elaborates that the optimal leverage classifiers achieve higher prediction accuracy with similar computing time under most scenarios. Compared to the full data approach, the proposed method yields significant computational time savings without sacrificing much accuracy. This aligns with our theoretical results that the convergence rate is only $O(\sqrt{N})$ while the computational cost is $O(N^3)$. In particular, leverage classifiers are more robust than logistic regression since the SVM only depends on the support vectors, while logistic regression is related to the likelihood of the full dataset. Linear discriminant analysis and quadratic discriminant analysis may work well because they are model-based classifiers requiring Gaussian distribution assumption.  Stochastic gradient descent algorithm can significantly reduce computational resources for large-scale datasets or online datastreams, but each iteration is updated by random sampling, which may lead to the loss of informative data points, and affect accuracy, particularly in imbalanced and mixed settings.  In Scenario IV, the prediction accuracy of our classifiers is about 10\% higher than others. Overall, it is promising that the leverage classifiers using a reduced dataset can outperform some classifiers using the full sample.
		
		\begin{figure}[htbp]
			\centering
			\setlength{\abovecaptionskip}{0cm}
			\setlength{\belowcaptionskip}{-0.cm}
			\includegraphics[width=1\textwidth]{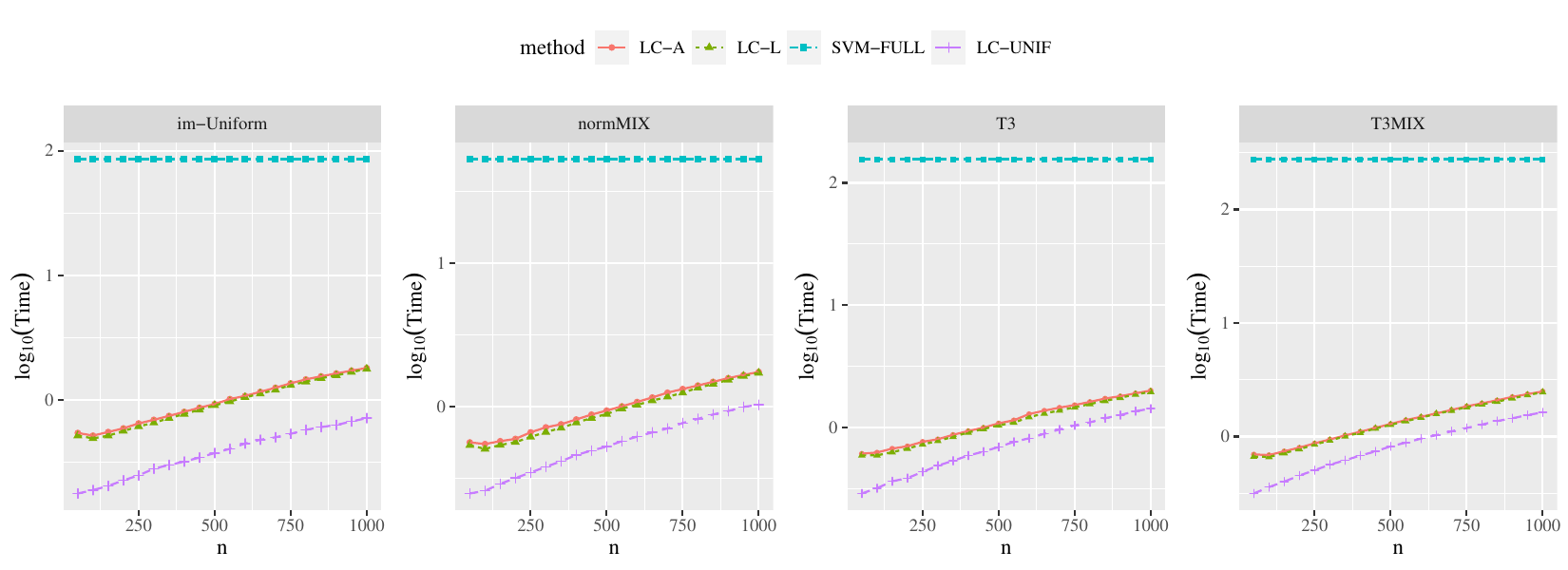}\par
			\vspace{-0.2cm}
			\caption{Comparison of CPU time (in seconds) against different subsample sizes
				under Scenarios I--IV. The logarithm is taken on time for a better presentation of the figures. 
			}\label{estimation-subsample-time}
		\end{figure}
		
		To validate the computational benefit of the leverage classifiers for large datasets, we further record the average computing time 
		for each method during 500 replications. We use the fast R package \texttt{LiblineaR} to fit the full sample SVM. Figure \ref{estimation-subsample-time} illustrates that the computing time of the full sample SVM is significantly larger than all leverage classifiers, as expected. our optimal leverage classifiers require slightly more time than uniform subsampling, this is due to the additional pilot study required to determine subsampling probabilities.  Moreover, due to additional calculations with the Hessian matrix in A-optimality, the L-optimal subsampling probabilities take less computing time than A-optimality, which is consistent with our computational complexity analysis in Section 3.2. Figure \ref{plane-estimation-MSE} and Figure \ref{estimation-subsample-time} both show that increasing $n$ leads to smaller MSE but also requires more computing time. The trade-off between estimation efficiency and computational efficiency actually affect by the practitioners' resource constraints and efficiency requirements, such as measurement cost, processing time, memory capacity, and prediction accuracy. We also report the computing time via one replication for different full sample sizes under Scenario I in Table~\ref{estimation-full-time}. The computational advantage of leverage classifiers becomes significant as $N$ increases.

		\begin{table}[htbp]
			\tabcolsep 6pt
			\caption{Comparison of CPU time (in seconds) under Scenario I when $n=1000$.}
			\label{estimation-full-time}
			\begin{center}
				\centering
				{
					\scalebox{1}{
						\begin{tabular}{lrrrrrrrrrrrrrrrrrrrrrrr}
							\hline
							\text{Method}&\multicolumn{13}{c}{$N$}\\
							\cline{2-14}
							&$ 10^3$&&$10^4$&&$10^5$&&$10^6$&&$10^7$\\
							\hline
							\textbf{LC-A}     &1.32&&1.33&&1.75&&1.85&&3.82\\
							\textbf{LC-L}     &1.32&&1.29&&1.48&&1.56&&2.62\\
							\textbf{LC-UNIF}  &0.29&&0.50&&0.53&&0.64&&0.69\\
							\textbf{SVM-FULL} &0.08&&0.65&&9.43&&240.48&&2526.90\\
							\hline
				\end{tabular}}}
			\end{center}
		\end{table}

		\section{Real Data Analysis}\label{CASP}
		
		Protein structure prediction is a critical challenge in computational biology \citep{lesk2019introduction}, and SVM has been a popular method for this task.  However, the high computational cost associated with SVM has limited its widespread applications in this field. 
		To this end, we examine the performance of our leverage classifier in protein structure prediction using the  ``Physicochemical Properties of Protein Tertiary Structure Dataset''. This dataset is taken from the critical assessment of protein structure prediction (CASP) experiments and includes 45,730 decoys with nine covariates.
		More details are available at the UCI machine learning repository \citep{Dua:2019}.
		
		\begin{figure}[htbp]
			\centering
			\setlength{\abovecaptionskip}{0cm}
			\setlength{\belowcaptionskip}{-0.cm}
			\includegraphics[width=1\textwidth]{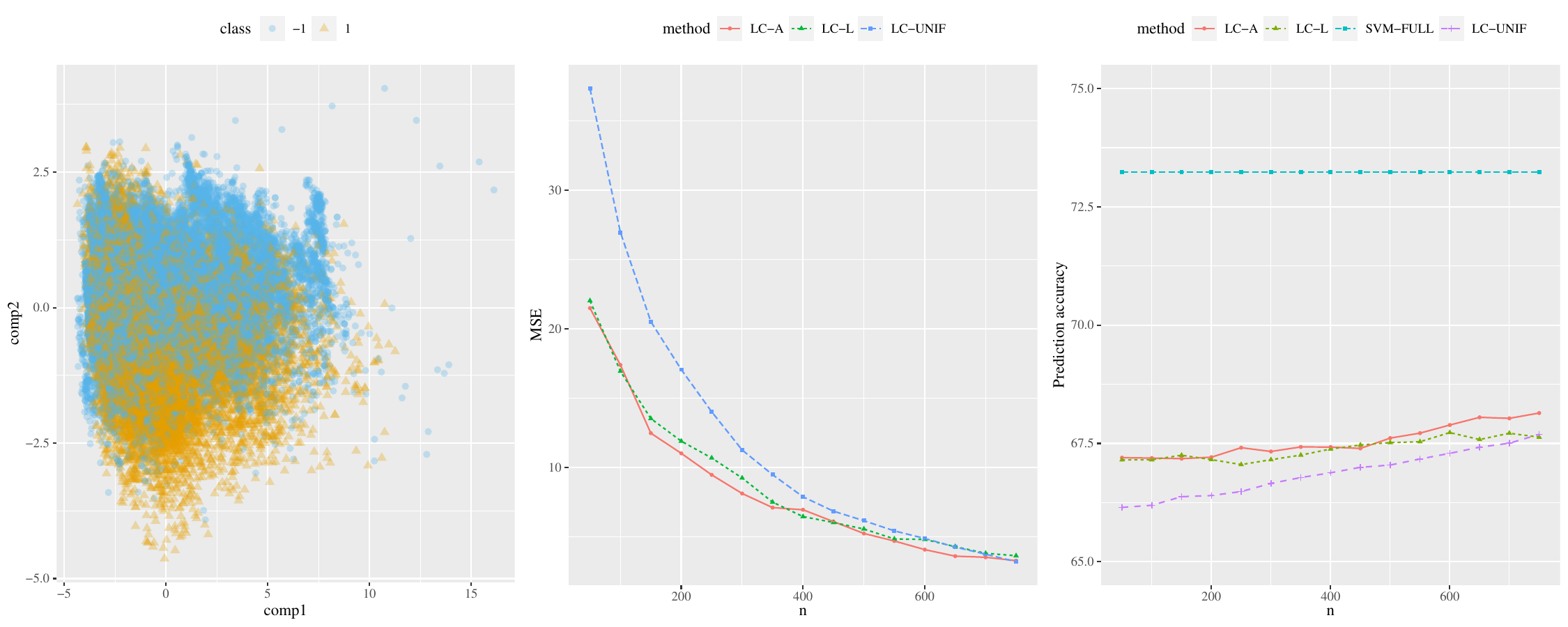}\par
			\vspace{-0.2cm}
			\caption{Analysis results for CASP dataset. Left panel: visualization with principal component analysis. Middle panel: MSE in approximating the full sample SVM estimator $\widehat{\bm\beta}$. Right panel: prediction accuracy (\%). }\label{CASP-plot}
		\end{figure}
		
		Root mean squared deviation (RMSD) is widely used as a metric for measuring the deviation of protein structures from their native protein structures  \citep{iraji2016rmsd}. In this analysis, our goal is to construct a classifier and predict whether the root mean squared deviation is greater than ten or not. This setting leads to the proportions of two classes about 40\% and 60\%.  Before applying our methods,
		we standardize each input variable with mean zero and standard deviation one, and then visualize it shown in the left panel of Figure \ref{CASP-plot}. We randomly select half of the dataset as the training set and leave the rest as the testing set for prediction. Uniformly choose $n_0=500$ pilot subsamples from the training set to obtain the subsampling probabilities $\widehat{{\bm\pi}}^{\text{A}}$ and $\widehat{{\bm\pi}}^{\text{L}}$.

		\begin{table}[htbp]      \tabcolsep 6pt
			\caption{Comparison of CPU time (in seconds) for CASP dataset.}
			\label{CASP-time}
			\begin{center}
				\centering
				{
					\scalebox{1}{
						\begin{tabular}{llrrrrrrrrrrrrrrrr}
							\hline
							\text{Method}&\multicolumn{9}{c}{$n$}\\
							\cline{2-10}
							& 50 & 100 & 200 &300 &400 &500&600&700&800 \\
							\hline
							\textbf{LC-A}&0.70&0.76&0.87&1.01&1.18&1.33&1.52&1.75&2.03\\
							\textbf{LC-L}&0.69&0.75&0.85&1.00&1.16&1.33&1.51&1.76&2.03\\
							\textbf{LC-UNIF}&0.39&0.42&0.53&0.66&0.81&0.96&1.11&1.20&1.52\\
							\textbf{SVM-FULL}&\multicolumn{8}{c}{11.93}\\
							\hline
				\end{tabular}}}
			\end{center}
		\end{table}
		
		In Table \ref{CASP-time}, our optimal leverage classifiers are significantly faster than the full sample SVM which is implemented by the fast R package \texttt{LiblineaR}. This phenomenon agrees with numerical studies, and it is a great improvement of the method to approximate the full sample SVM. 
		The middle and right panels of Figure \ref{CASP-plot} present the mean squared errors of approximating the full sample SVM estimator and the prediction performances. The good performance of the optimal leverage classifiers is consistent with our theory and the numerical studies.
		
		\section{Conclusion}\label{conclusion}

		Constructing accurate classifiers with informative subsamples from large-scale datasets is a crucial task in statistical analysis and machine learning. In this paper, we propose a novel leverage classifier for SVM under the subsampling framework to address the computational challenge. We construct optimal leverage classifiers by minimizing the unconditional asymptotic variance with double randomnesses.  Our extensive numerical investigations demonstrate that the proposed methods provide satisfactory performances in estimation, computation, and prediction.
		
		Subsampling is a fast and effective strategy for processing large-scale datasets and further research is needed for more delicate statistical models. We conclude this paper with several future topics.  First, our binary subsampling leverage classifier may be extended to multi-classification problems by one-versus-one or one-versus-rest SVM in a linear nonseparable setting. Second, one limitation in our work is that we only focus on the linear SVM for nonseparable cases to shed light on the leverage classifiers. Extensions of the leverage classifiers to more general settings, such as kernel SVM in reproducing kernel Hilbert spaces, remain challenging because it is unclear how to integrate existing asymptotic results \citep{hable2012asymptotic} with our subsampling framework. Third, it is worth further exploring the trade-off between estimation efficiency and computation complexity under measurement constraints. Finally, investigating other optimal criteria, such as minimizing the classification error or maximizing the prediction accuracy, also merits further research.

		
%
		
		\vspace*{-10pt}
		\par

		\section*{Supplementary Material~~}
		\label{SM}
		A supplementary PDF file contains the proof of theoretical results and additional simulation results in our paper.

		
%
		
		\bibliographystyle{apalike}
		\bibliography{SVM_Sampling_references}

\begin{thebibliography}{}

\bibitem[Ai et~al., 2018]{ai2018optimal}
Ai, M., Yu, J., Zhang, H., and Wang, H. (2018).
\newblock Optimal subsampling algorithms for big data regressions.
\newblock {\em Statistica Sinica}, \textbf{6}(4):363--392.

\bibitem[Atkinson et~al., 2007]{atkinson2007optimum}
Atkinson, A., Donev, A., and Tobias, R. (2007).
\newblock {\em Optimum experimental designs, with SAS}.
\newblock Oxford University Press.

\bibitem[Bordes et~al., 2005]{bordes2005fast}
Bordes, A., Ertekin, S., Weston, J., and Bottou, L. (2005).
\newblock Fast kernel classifiers with online and active learning.
\newblock {\em Journal of Machine Learning Research}, \textbf{6}(9):1579--1619.

\bibitem[Boser et~al., 1992]{boser1992training}
Boser, B.~E., Guyon, I.~M., and Vapnik, V.~N. (1992).
\newblock A training algorithm for optimal margin classifiers.
\newblock In {\em Proceedings of the Fifth Annual Workshop on Computational
  Learning Theory}, pages 144--152.

\bibitem[Camelo et~al., 2015]{camelo2015nearest}
Camelo, S.~A., Gonz{\'a}lez-Lima, M.~D., and Quiroz, A. (2015).
\newblock Nearest neighbors methods for support vector machines.
\newblock {\em Annals of Operations Research}, \textbf{235}(1):85--101.

\bibitem[Chang, 2011]{chang2011psvm}
Chang, E.~Y. (2011).
\newblock {PSVM}: Parallelizing support vector machines on distributed
  computers.
\newblock In {\em Foundations of Large-Scale Multimedia Information Management
  and Retrieval}, pages 213--230. Springer.

\bibitem[Cortes and Vapnik, 1995]{cortes1995support}
Cortes, C. and Vapnik, V. (1995).
\newblock Support-vector networks.
\newblock {\em Machine Learning}, \textbf{20}(3):273--297.

\bibitem[Drineas et~al., 2012]{drineas2012fast}
Drineas, P., Magdon-Ismail, M., Mahoney, M.~W., and Woodruff, D.~P. (2012).
\newblock Fast approximation of matrix coherence and statistical leverage.
\newblock {\em Journal of Machine Learning Research},
  \textbf{13}(1):3475--3506.

\bibitem[Drineas et~al., 2011]{drineas2011faster}
Drineas, P., Mahoney, M.~W., Muthukrishnan, S., and Sarl{\'o}s, T. (2011).
\newblock Faster least squares approximation.
\newblock {\em Numerische Mathematik}, \textbf{117}(2):219--249.

\bibitem[Dua and Graff, 2017]{Dua:2019}
Dua, D. and Graff, C. (2017).
\newblock {UCI} machine learning repository.

\bibitem[Durrett, 2019]{durrett2019probability}
Durrett, R. (2019).
\newblock {\em Probability: theory and examples}, volume~49.
\newblock Cambridge university press.

\bibitem[Fan et~al., 2020]{fan2020statistical}
Fan, J., Li, R., Zhang, C.-H., and Zou, H. (2020).
\newblock {\em Statistical foundations of data science}.
\newblock Chapman and Hall/CRC.

\bibitem[Hable, 2012]{hable2012asymptotic}
Hable, R. (2012).
\newblock Asymptotic normality of support vector machine variants and other
  regularized kernel methods.
\newblock {\em Journal of Multivariate Analysis}, {\bf 106}:92--117.

\bibitem[Han et~al., 2023]{Han2023model}
Han, Y., Ma, P., Ren, H., and Wang, Z. (2023).
\newblock Model checking in large-scale dataset via
  structure-adaptive-sampling.
\newblock {\em Statistica Sinica}, {\textbf{33}}:303--329.

\bibitem[Hastie et~al., 2010]{hastie2010elements}
Hastie, T., Tibshirani, R., and Friedman, J. (2010).
\newblock {\em The elements of statistical learning: data mining, inference,
  and prediction}.
\newblock Springer Science \& Business Media.

\bibitem[Hsieh et~al., 2008]{hsieh2008dual}
Hsieh, C.-J., Chang, K.-W., Lin, C.-J., Keerthi, S.~S., and Sundararajan, S.
  (2008).
\newblock A dual coordinate descent method for large-scale linear {SVM}.
\newblock In {\em Proceedings of the 25th International Conference on Machine
  Learning}, pages 408--415.

\bibitem[Iraji and Ameri, 2016]{iraji2016rmsd}
Iraji, M.~S. and Ameri, H. (2016).
\newblock {RMSD} protein tertiary structure prediction with soft computing.
\newblock {\em IJ Mathematical Sciences and Computing}, {\textbf 2}:24--33.

\bibitem[Kaufman, 1998]{kaufman1998solving}
Kaufman, L. (1998).
\newblock Solving the quadratic programming problem arising in support vector
  classification.
\newblock {\em Advances in Kernel Methods-Support Vector Learning}, pages
  147--167.

\bibitem[Kimeldorf and Wahba, 1971]{kimeldorf1971some}
Kimeldorf, G. and Wahba, G. (1971).
\newblock Some results on {T}chebycheffian spline functions.
\newblock {\em Journal of Mathematical Analysis and Applications},
  \textbf{33}(1):82--95.

\bibitem[Koo et~al., 2008]{koo2008bahadur}
Koo, J.-Y., Lee, Y., Kim, Y., and Park, C. (2008).
\newblock A {B}ahadur representation of the linear support vector machine.
\newblock {\em Journal of Machine Learning Research}, \textbf{9}(7):1343--1368.

\bibitem[Lesk, 2019]{lesk2019introduction}
Lesk, A. (2019).
\newblock {\em Introduction to Bioinformatics}.
\newblock Oxford University Press.

\bibitem[Li and Meng, 2020]{li2020modern}
Li, T. and Meng, C. (2020).
\newblock Modern subsampling methods for large-scale least squares regression.
\newblock {\em International Journal of Cyber-Physical Systems (IJCPS)},
  2(2):1--28.

\bibitem[Lin, 2004]{lin2004note}
Lin, Y. (2004).
\newblock A note on margin-based loss functions in classification.
\newblock {\em Statistics \& Probability letters}, \textbf{68}(1):73--82.

\bibitem[Lin et~al., 2002]{lin2002statistical}
Lin, Y., Wahba, G., Zhang, H., and Lee, Y. (2002).
\newblock Statistical properties and adaptive tuning of support vector
  machines.
\newblock {\em Machine Learning}, \textbf{48}(1):115--136.

\bibitem[Ma et~al., 2022]{ma2022asymptotic}
Ma, P., Chen, Y., Zhang, X., Xing, X., Ma, J., and Mahoney, M.~W. (2022).
\newblock Asymptotic analysis of sampling estimators for randomized numerical
  linear algebra algorithms.
\newblock {\em Journal of Machine Learning Research}, {\bf 23}(1):7970--8014.

\bibitem[Ma et~al., 2015a]{ma2015efficient}
Ma, P., Huang, J.~Z., and Zhang, N. (2015a).
\newblock Efficient computation of smoothing splines via adaptive basis
  sampling.
\newblock {\em Biometrika}, \textbf{102}(3):631--645.

\bibitem[Ma et~al., 2015b]{ma2015statistical}
Ma, P., Mahoney, M.~W., and Yu, B. (2015b).
\newblock A statistical perspective on algorithmic leveraging.
\newblock {\em Journal of Machine Learning Research}, \textbf{16}(1):861--911.

\bibitem[Mahoney and Drineas, 2009]{mahoney2009cur}
Mahoney, M.~W. and Drineas, P. (2009).
\newblock {CUR} matrix decompositions for improved data analysis.
\newblock {\em Proceedings of the National Academy of Sciences},
  \textbf{106}(3):697--702.

\bibitem[Mehrotra, 1992]{mehrotra1992implementation}
Mehrotra, S. (1992).
\newblock On the implementation of a primal-dual interior point method.
\newblock {\em SIAM Journal on Optimization}, \textbf{2}(4):575--601.

\bibitem[Meng et~al., 2021]{meng2021lowcon}
Meng, C., Xie, R., Mandal, A., Zhang, X., Zhong, W., and Ma, P. (2021).
\newblock Lowcon: A design-based subsampling approach in a misspecified linear
  model.
\newblock {\em Journal of Computational and Graphical Statistics}, \textbf
  {30}(3):694--708.

\bibitem[Meng et~al., 2020]{meng2020more}
Meng, C., Zhang, X., Zhang, J., Zhong, W., and Ma, P. (2020).
\newblock More efficient approximation of smoothing splines via space-filling
  basis selection.
\newblock {\em Biometrika}, \textbf{107}(3):723--735.

\bibitem[Ohlsson, 1989]{ohlsson1989asymptotic}
Ohlsson, E. (1989).
\newblock Asymptotic normality for two-stage sampling from a finite population.
\newblock {\em Probability Theory and Related Fields}, \textbf{81}(3):341--352.

\bibitem[Platt, 1998]{platt1998fast}
Platt, J. (1998).
\newblock Fast training of support vector machines using sequential minimal
  optimization.
\newblock In {\em Advances in Kernel Methods - Support Vector Learning}. MIT
  Press.

\bibitem[Pollard, 1991]{pollard1991asymptotics}
Pollard, D. (1991).
\newblock Asymptotics for least absolute deviation regression estimators.
\newblock {\em Econometric Theory}, \textbf{7}(2):186--199.

\bibitem[Ren et~al., 2022]{ren2022large}
Ren, H., Zou, C., Chen, N., and Li, R. (2022).
\newblock Large-scale datastreams surveillance via pattern-oriented-sampling.
\newblock {\em Journal of the American Statistical Association}, {\bf
  117}(538):794--808.

\bibitem[Sch{\"o}lkopf et~al., 2001]{scholkopf2001generalized}
Sch{\"o}lkopf, B., Herbrich, R., and Smola, A.~J. (2001).
\newblock A generalized representer theorem.
\newblock In {\em International Conference on Computational Learning Theory},
  pages 416--426. Springer.

\bibitem[Scott and Terrell, 1987]{scott1987biased}
Scott, D.~W. and Terrell, G.~R. (1987).
\newblock Biased and unbiased cross-validation in density estimation.
\newblock {\em Journal of the American Statistical Association}, {\bf
  82}(400):1131--1146.

\bibitem[Shalev-Shwartz et~al., 2011]{shalev2011pegasos}
Shalev-Shwartz, S., Singer, Y., Srebro, N., and Cotter, A. (2011).
\newblock Pegasos: Primal estimated sub-gradient solver for {SVM}.
\newblock {\em Mathematical Programming}, \textbf{127}(1):3--30.

\bibitem[Sheather and Jones, 1991]{sheather1991reliable}
Sheather, S.~J. and Jones, M.~C. (1991).
\newblock A reliable data-based bandwidth selection method for kernel density
  estimation.
\newblock {\em Journal of the Royal Statistical Society: Series B
  (Methodological)}, {\bf 53}(3):683--690.

\bibitem[Silverman, 1986]{silverman1986density}
Silverman, B.~W. (1986).
\newblock {\em Density estimation for statistics and data analysis}.
\newblock Champman \& Hall.

\bibitem[Steinwart and Christmann, 2008]{steinwart2008support}
Steinwart, I. and Christmann, A. (2008).
\newblock {\em Support vector machines}.
\newblock Springer Science \& Business Media.

\bibitem[Tsang et~al., 2005]{tsang2005core}
Tsang, I.~W., Kwok, J.~T., and Cheung, P.-M. (2005).
\newblock Core vector machines: Fast {SVM} training on very large data sets.
\newblock {\em Journal of Machine Learning Research}, \textbf{6}(4):363--392.

\bibitem[Vapnik, 2013]{vapnik2013nature}
Vapnik, V. (2013).
\newblock {\em The Nature of Statistical Learning Theory}.
\newblock Springer Science \& Business Media.

\bibitem[Wahba et~al., 2003]{wahba2003optimal}
Wahba, G., Lin, Y., Lee, Y., and Zhang, H. (2003).
\newblock Optimal properties and adaptive tuning of standard and nonstandard
  support vector machines.
\newblock In {\em Nonlinear Estimation and Classification}, pages 129--147.
  Springer.

\bibitem[Wang and Ma, 2021]{wang2021optimal}
Wang, H. and Ma, Y. (2021).
\newblock Optimal subsampling for quantile regression in big data.
\newblock {\em Biometrika}, \textbf{108}(1):99--112.

\bibitem[Wang et~al., 2018]{wang2018optimal}
Wang, H., Zhu, R., and Ma, P. (2018).
\newblock Optimal subsampling for large sample logistic regression.
\newblock {\em Journal of the American Statistical Association},
  \textbf{113}(522):829--844.

\bibitem[Wang et~al., 2012]{wang2012breaking}
Wang, Z., Crammer, K., and Vucetic, S. (2012).
\newblock Breaking the curse of kernelization: Budgeted stochastic gradient
  descent for large-scale svm training.
\newblock {\em Journal of Machine Learning Research},
  \textbf{13}(1):3103--3131.

\bibitem[Williams and Seeger, 2000]{williams2001using}
Williams, C. and Seeger, M. (2000).
\newblock Using the nystr{\"o}m method to speed up kernel machines.
\newblock {\em Advances in Neural Information Processing Systems}, \textbf{13}.

\bibitem[Yu et~al., 2022]{yu2022optimal}
Yu, J., Wang, H., Ai, M., and Zhang, H. (2022).
\newblock Optimal distributed subsampling for maximum quasi-likelihood
  estimators with massive data.
\newblock {\em Journal of the American Statistical Association},
  \textbf{117}(537):265--276.

\bibitem[Zhan, 2004]{zhan2004matrix}
Zhan, X. (2004).
\newblock {\em Matrix Inequalities}.
\newblock Springer.

\bibitem[Zhang et~al., 2021]{zhang2021optimal}
Zhang, T., Ning, Y., and Ruppert, D. (2021).
\newblock Optimal sampling for generalized linear models under measurement
  constraints.
\newblock {\em Journal of Computational and Graphical Statistics},
  \textbf{30}(1):106--114.

\end{thebibliography}


\newpage
\fancyhead{}
\pagestyle{plain}
\setcounter{page}{1}


	\centerline{\large\bf  Supplementary Material for ``LEVERAGE CLASSIFIER: ANOTHER LOOK AT}
\vspace{6pt}
\centerline{\large\bf SUPPORT VECTOR MACHINE"}
\vspace{.25cm}

\author{Author(s)}
\vspace{.4cm}
\centerline{Yixin Han\textsuperscript{1}, Jun Yu\textsuperscript{2}, Nan Zhang\textsuperscript{3}, Cheng Meng\textsuperscript{4}, Ping Ma\textsuperscript{5}, Wenxuan Zhong\textsuperscript{5}, and Changliang Zou\textsuperscript{1}} \vspace{.4cm}

\centerline{\it \textsuperscript{1}School of Statistics and Data Science, LPMC $\&$ KLMDASR, Nankai University, Tianjin, P.R. China}

\centerline{\it\textsuperscript{2}School of Mathematics and Statistics, Beijing Institute of Technology, Beijing, P.R.China}

\centerline{\it \textsuperscript{3}School of Data Science, Fudan University, Shanghai, P.R.China}

\centerline{\it  \textsuperscript{4}Institute of Statistics and Big Data, Renmin University, Beijing, P.R.China}

\centerline{\it \textsuperscript{5}Department of Statistics, University of Georgia, Athens, GA, USA}

\vspace{.55cm} \fontsize{9}{11.5pt plus.8pt minus
	.6pt}\selectfont


\def\thelemma{S.\arabic{lemma}}
\def\thepro{S.\arabic{pro}}
\def\theequation{S.\arabic{equation}}
\def\thetable{S\arabic{table}}
\def\thefigure{S\arabic{figure}}
\setcounter{figure}{0}
\setcounter{table}{0}

\fontsize{12}{14pt plus.8pt minus .6pt}\selectfont

This supplementary material contains the proofs of technical results and some additional simulation results.

\beginsupplement
\renewcommand{\thesection}{}
\setcounter{Lem}{0}
\renewcommand{\theLem}{A.\arabic{Lem}}

\renewcommand{\theequation}{}
\renewcommand{\theequation}{S.\arabic{equation}}

	\beginsupplement
\renewcommand{\thesection}{}
\setcounter{Lem}{0}
\renewcommand{\theLem}{A.\arabic{Lem}}

\renewcommand{\theequation}{}
\renewcommand{\theequation}{S.\arabic{equation}}

\appendix
\section*{Appendix A: Useful Lemma}
\renewcommand{\thesection}{}
\setcounter{Lem}{0}
\renewcommand{\theLem}{S.\arabic{Lem}}

The following Lemma is a multivariate extension of the martingale central limit theorem, see Lemma 4 in \citet{zhang2021optimal} for details.

\begin{Lem}[Multivariate version of martingale CLT]\label{martingaleCLT}
	Let $\left\{{\bm\eta}_{ki},i=1,\ldots, N_k\right\}$ be a martingale difference sequence in $\mathbb{R}^{p}$ relative to the filtration $\left\{\mathcal{F}_{ki},i=0,1,\ldots, N_k\right\}$ and let ${\bm Z}_k\in\mathbb{R}^{p}$ be an $\mathcal{F}_{k0}$-measurable random vector for $k=1,2,3,\ldots$. Denote $ \bm R_k=\sum\nolimits_{i=1}^{N_k}\bm\eta_{ki}$. Assume the following conditions hold.
	\begin{itemize}
		\item [(i)] $\lim\nolimits_{k \to \infty}\sum\nolimits_{i=1}^{N_k}{\E}\left(\|\bm\eta_{ki}\|^4\right)=0$.
		\item [(ii)] $\lim\nolimits_{k\to\infty}{\E}\left\{\|\sum\nolimits_{i=1}^{N_k}{\E}\left(\bm\eta_{ki}\bm\eta_{ki}^{\top}\mid\mathcal{F}_{k,i-1}\right)-{ \bf B}_k\|^2\right\}=0$ for some sequence of positive-definite matrices $\left\{{ \bf B}_k\right\}_{k=1}^\infty$ with $\sup\nolimits_{k}\lambda_{{\max}}({ \bf B}_k)<\infty$, say that the largest eigenvalue is uniformly bounded.
		\item [(iii)] For a probability distribution $\bm L_0$, $*$ denotes convolution and $\bm L(\cdot)$ denotes the {law} of random variables, $\bm L( \bm Z_k)*\mathcal{N}({\bf 0},{\bf B}_k){\rightarrow} \bm L_0$, where the convergence is in distribution.
	\end{itemize}
	Then we have
	\begin{align*}
		\bm L(\bm Z_k+\bm R_k){\rightarrow} \bm L_0.
	\end{align*}
\end{Lem}

\section*{Appendix B: Proof of Theorem \ref{Bahadur-type representation}}
\begin{proof}
	Denote
	\[
	L_n(\bm\beta)= \frac{1}{n}\sum\limits_{i=1}^n\frac{1}{N\pi_i^*}\left[1-Y_i^*f({{\X}}_i^*,\bm\beta)\right]_{+}, L_N(\bm\beta)=\frac{1}{N}\sum\limits_{j=1}^N\left[1-Y_jf({{\X}}_j,\bm\beta)\right]_{+},
	\]
	\[
	l_{\lambda,n}(\bm\beta)=\frac{1}{n}\sum\limits_{i=1}^n\frac{1}{N\pi_i^*}\left[1-Y_i^*f({{\X}}_i^*,\bm\beta)\right]_{+}+\frac{\lambda}{2}\|\bm\beta_1\|^2.
	\]

	The proof can be divided into the following intermediate parts.
	
	First, we consider the influence of a fixed $\lambda$.
	For a fixed $\bm\theta=(1,\bm\theta_1^\top)^\top\in\mathbb{R}^{p+1}$, define
	\begin{align*}
		\Lambda_n({\bm\theta})=n\left\{l_{\lambda,n}\left({\bm\beta}^{\dagger}+\frac{\bm\theta}{\sqrt{n}}\right)-l_{\lambda,n}\left({\bm\beta}^{\dagger}\right)\right\},~~T_n(\bm\theta)=\mathbb{E}\left\{\Lambda_n(\bm\theta)\right\}.
	\end{align*}
	
	Observe that
	\begin{align*}
		\Lambda_n(\bm\theta) =& \sum\limits_{i=1}^n\frac{1}{N\pi_i^*}\left\{\left[1-Y_i^*f\left({{\X}}_i^*,{\bm\beta}^{\dagger}+\frac{\bm\theta}{\sqrt{n}}\right)\right]_+
		- \left[1-Y_i^*f({{\X}}_i^*,{\bm\beta}^{\dagger})\right]_+\right\}\\
		&+n\frac{\lambda}{2}\left(\|{\bm\beta}^{\dagger}_{1} +\frac{\bm\theta_1}{\sqrt{n}}\|^2-\|{\bm\beta}^{\dagger}_{1}\|^2\right),
	\end{align*}
	and $
	{\E}\left\{L_n(\bm\beta)\right\} = {\E}\left[{\E}\left\{L_n(\bm\beta)\mid \mathcal{D}_N\right\}\right] = L(\bm\beta)=  {\E}\left[1-Yf({{\X}},\bm\beta)\right]_+
	$. Under Assumption \ref{derivativeHS}, we assume  $\bm\beta_{1}^{\dagger}\neq 0$ without loss of generality. By Lemma 3 in \citet{koo2008bahadur}, we have 
	\begin{align*}
		T_n(\bm\theta) &= n\left\{L\left({\bm\beta}^{\dagger}+\frac{\bm\theta}{\sqrt{n}}\right)-L({\bm\beta}^{\dagger})\right\}+\frac{\lambda}{2}\left(\|\bm\theta_1\|^2+2\sqrt{n}\bm\theta_1^\top{\bm\beta}^{\dagger}_{1}\right),\\
		&=\frac{1}{2}\bm\theta^\top {\bf H}(\breve{\bm\beta})\bm\theta + \frac{\lambda}{2}\left(\|\bm\theta_1\|^2+2\sqrt{n}\bm\theta_1^\top{\bm\beta}^{\dagger}_{1}\right),
	\end{align*}
	by applying Taylor expansion of $L(\bm\beta)$ around ${\bm\beta}^{\dagger}$,
	where $\breve{\bm\beta}={\bm\beta}^{\dagger}+(\bm\theta/\sqrt{n})t$ for some $0< t < 1$.
	
	Define ${\bf D}_{ij}(\bm\alpha)={ \bf H}({\bm\beta}^{\dagger}+\bm\alpha)_{ij}-{ \bf H}({\bm\beta}^{\dagger})_{ij}$ for $0\leq i,j\leq p+1$. By Assumption \ref{density}, ${\bf H}({\bm\beta})$ is continuous in ${\bm\beta}$. Then, for any $\varepsilon_1>0$, there exist $\delta_1>0$ such that ${\bf D}_{ij}(\bm\alpha)<\varepsilon_1$ if $\|\bm\alpha\|<\delta_1$ for all $0\leq i,j\leq p+1$. Thus, for sufficiently large $n$ such that $\|(\bm\theta/\sqrt{n})t\|<\delta_1$
	\begin{align*}
		\left|\bm\theta^\top\left({\bf H}(\breve{\bm\beta})-{\bf H}({\bm\beta}^{\dagger})\right)\bm\theta\right|\leq \sum\limits_{i,j}|\bm\theta_i||\bm\theta_j|\left|{\bf D}_{ij}\left(\frac{\bm\theta}{\sqrt{n}}t\right)\right|\leq 2\varepsilon_1\|\bm\theta\|^2,
	\end{align*}
	then $\bm\theta^\top { \bf H}(\breve{\bm\beta})\bm\theta/2=\bm\theta^\top {\bf H}({\bm\beta}^{\dagger})\bm\theta/2 + o(1)$ as $n\to\infty$. Combining the assumption that $\lambda = o(n^{-1/2})$, we have
	\begin{align*}
		T_n(\bm\theta) = \frac{1}{2}\bm\theta^\top {\bf H}({\bm\beta}^{\dagger})\bm\theta + o(1).
	\end{align*}
	
	Next, we would like to provide an expansion of $\Lambda_n(\bm\theta)$ {under Assumptions \ref{density}--\ref{derivativeHS}.} Let
	${\bm W}_n = - n^{-1} \sum\nolimits_{i=1}^n\left({N\pi_i^*}\right)^{-1}\xi_i^*Y_i^*\widetilde{{{\X}}}_i^*$, where $\xi_i^*={\I}\left(Y_i^*f({{\X}}_i^*,{\bm\beta}^{\dagger})\leq 1\right)$.
	If we define
	\begin{align*}
		R_{i,n}(\bm\theta)&=\frac{1}{N\pi_i^*}\left\{\left[1-Y_i^*f\left({{\X}}_i^*,{\bm\beta}^{\dagger}+\frac{\bm\theta}{\sqrt{n}}t\right)\right]_+-\left[1-Y_i^*f\left({{\X}}_i^*,{\bm\beta}^{\dagger}\right)\right]_++\xi_i^*Y_i^*f\left({{\X}}_i^*,\frac{\bm\theta}{\sqrt{n}}\right)\right\},\\
		R_{j,N}(\bm\theta)&=\left[1-Y_jf\left({{\X}}_j,{\bm\beta}^{\dagger}+\frac{\bm\theta}{\sqrt{n}}t\right)\right]_+-\left[1-Y_jf\left({{\X}}_j,{\bm\beta}^{\dagger}\right)\right]_++\xi_jY_jf\left({{\X}}_j,\frac{\bm\theta}{\sqrt{n}}\right),
	\end{align*}
	where $i=1,\ldots,n$ and $j=1,\ldots,N$. Recall that ${\E}\{\left(N\pi_i^*\right)^{-1}\xi_i^*Y_i^*\widetilde{ \X}_i^*\}={\bm S}(\bm\beta^\dagger)=0$. Recall the definitions of $T_n(\bm\theta)$ and $\bm W_n$, we have
	\begin{align}\label{eq:s1}
		\Lambda_n(\bm\theta)=&\sum\limits_{i=1}^n\frac{1}{N\pi_i^*}\left[1-Y_i^*f\left(\bm{X}_i^*,\bm\beta^\dagger+\frac{\bm\theta}{\sqrt{n}}\right)\right]_+ -nL\left(\bm\beta^\dagger+\frac{\bm\theta}{\sqrt{n}}\right)\nonumber\\
		&-\sum\limits_{i=1}^n\frac{1}{N\pi_i^*}\left[1-Y_i^*f\left(\bm{X}_i^*,\bm\beta^\dagger\right)\right]_+ +nL\left(\bm\beta^\dagger\right)+\frac{\lambda}{2}\left(\|\bm\theta_1\|^2+2\sqrt{n}\bm\theta_1^\top{\bm\beta}^{\dagger}_{1}\right)\nonumber\\
		&+\sum\limits_{i=1}^n\frac{1}{N\pi_i^*}\xi_i^*Y_i^*(\widetilde{\bm{X}}_i^{*})^{\top}\frac{\bm\theta}{\sqrt{n}}-\sum\limits_{i=1}^n\frac{1}{N\pi_i^*}\xi_i^*Y_i^*(\widetilde{\bm{X}}_i^{*})^{\top}\frac{\bm\theta}{\sqrt{n}}\nonumber\\
		=& T_n(\bm\theta)+\sqrt{n}{\bm W}_n^{\top}\bm\theta+\sum\limits_{i=1}^n\left[R_{i,n}(\bm\theta)-\mathbb{E}\left\{R_{i,n}(\bm\theta)\right\}\right].
	\end{align}
	
	Recall that $\left[\cdot\right]_+$ denotes the hinge loss. We define $\varphi ={\I}\left(a\leq 1\right)$ and $D=\left[1-z\right]_+-\left[1-a\right]_++\varphi(z-a)$. 
	Then we have
	\begin{align}\label{tailtermprof}
		\begin{split}
			D &= (1-z){\I}(a>1,z\leq 1)+(z-1){\I}(a<1,z>1)\\
			& \leq \left|z-a\right|{\I}(a>1,z\leq 1) +\left|z-a\right|{\I}(a<1,z>1)\\
			&=\left|z-a\right|\left\{{\I}(a>1,z\leq 1)+{\I}(a<1,z>1)\right\}\\
			&\leq \left|z-a\right|{\I}\left(\left|1-a\right|\leq \left|z-a\right|\right).
		\end{split}
	\end{align}
	
	Let $z_i=Y_i^*f({\X}_i^*,\bm\beta^{\dagger}+\bm\theta/\sqrt{n})$ and $a_i=Y_i^*f({\X}_i^*,\bm\beta^\dagger)$ in \eqref{tailtermprof}, we have 
	\begin{align}\label{tailterm}
		\begin{split}
			\left|R_{i,n}(\bm\theta)\right|&\leq \frac{1}{N\pi_i^*}\left|\frac{f({{\X}}_i^*,\bm\theta)}{\sqrt{n}}\right|U_i\left(\left|\frac{f({{\X}}_i^*,\bm\theta)}{\sqrt{n}}\right|\right),
		\end{split}
	\end{align}
	where $U_i(t)={\I}\left(\left|1-Y_i^*f({{{\X}}_i^*},{\bm\beta}^{\dagger})\right|\leq t\right)$ with respect to the $i$-th subsample point for $t\in\mathbb{R}$.  By \eqref{tailterm}, for each fixed $\bm\theta$ we obtain
	\begin{align*}
		{\E}\left[\sum\limits_{i=1}^n\left\{R_{i,n}(\bm\theta)-{\E}\left(R_{i,n}(\bm\theta)\right)\right\}\right]^2
		&={\E}\left\{{\E}\left[\sum\limits_{i=1}^n\left\{R_{i,n}(\bm\theta)-{\E}\left(R_{i,n}(\bm\theta)\right)\right\}\right]^2\bigg{|} \mathcal{D}_N\right\}\nonumber\\
		&=\frac{n}{N^2}\sum\limits_{j=1}^N{\E}\left[\frac{1}{\pi_j}\left\{R_{j,N}(\bm\theta)-{\E}\left(R_{j,N}(\bm\theta)\right)\right\}^2\right]\nonumber\\
		&\leq \frac{n}{N^2}\sum\limits_{j=1}^N{\E}\left\{\frac{1}{\pi_j}R_{i,N}^2(\bm\theta)\right\}\\
		&\leq \frac{n}{N^2}\sum\limits_{j=1}^N{\E}\left\{\frac{1}{\pi_j}\left(1+\|{{\X}}_j\|^2\right)\frac{\|\bm\theta\|^2}{n}U_j\left(\sqrt{1+\|{{\X}}_j\|^2}\frac{\|\bm\theta\|}{\sqrt{n}}\right)\right\}\nonumber\\
		&\leq \frac{\|\bm\theta\|^2}{N^2}\sum\limits_{j=1}^N{\E}\left\{\frac{1}{\pi_j}\left(1+\|{{\X}}_j\|^2\right)U_j\left(\sqrt{1+\|{{\X}}_j\|^2}\frac{\|\bm\theta\|}{\sqrt{n}}\right)\right\}.\nonumber
	\end{align*}
	
	By Assumption \ref{density} implies that ${\E}(\|{{\X}}\|^4)<\infty$, there exists $c_1$ such that \[
	{\E}\left\{(1+\|{\X}\|^4){\I}\left(\|{{\X}}\|>c_1\right)\right\}<\varepsilon_2/2,
	\]
	for any $\varepsilon_2>0$. Let $U(t)={\I}\left(\left|1-Yf({{{\X}}},{\bm\beta}^{\dagger})\right|\leq t\right)$ for $t\in\mathbb{R}$. By Assumption \ref{approxexpansion} and holder inequality, we have
	\begin{align*}
		&\frac{1}{N^2}\sum\limits_{j=1}^N{\E}\left\{\frac{1}{\pi_j}\left(1+\|{{\X}}_j\|^2\right)U_j\left(\sqrt{1+\|{{\X}}_j\|^2}\frac{\|\bm\theta\|}{\sqrt{n}}\right)\right\}\\
		\leq&\frac{1}{N^2}\sum\limits_{j=1}^N{\E}\left\{\frac{1}{{\pi}_j}\left(1+\|{{\X}}_j\|^2\right){\I}\left(\|{{\X}}_j\|>c_1\right)\right\}+\frac{1}{N^2}\sum\limits_{j=1}^N{\E}\left\{\frac{1+c_1^2}{{\pi}_j}U\left(\sqrt{1+c_1^2}\frac{\|\bm\theta\|}{\sqrt{n}}\right)\right\}\\
		\leq&\sqrt{\E \left(\frac{1}{N^3}\sum_{j=1}^N \frac{1}{\pi_j^2}\right)}\sqrt{\E\left\{\frac{1}{N}\sum_{j=1}^N\left(1+\|{{\X}}_j\|^2\right)^2{\I}\left(\|{{\X}}_j\|>c_1\right)\right\}}\\
		&+(1+c_1^2)\sqrt{\E\left( \frac{1}{N^3}\sum_{j=1}^N \frac{1}{\pi_j^2}\right)}\sqrt{\frac{1}{N}\sum_{j=1}^N {\P}\left\{U\left(\sqrt{1+c_1^2}{\|\bm\theta\|}/{\sqrt{n}}\right)=1\right\}},
	\end{align*}
	
	By Assumption \ref{density}, the conditional distribution of ${\X}$ given $Y$ is not degenerate, which implies $\lim_{t\to 0}\P\left(U(t)=1\right)=0$. We can take a large $c_2$ such that \[
	\P\left\{U\left(\sqrt{1+c_1^2}{\|\bm\theta\|}/{\sqrt{n}}\right)=1\right\} < \varepsilon_2/\left\{2(1+c_1^2)\right\},
	\]
	for $n>c_2$. By Assumption \ref{approxexpansion}, it proves that
	$
	{\E}\left[\sum\nolimits_{i=1}^n\left\{R_{i,n}(\bm\theta)-{\E}\left(R_{i,n}(\bm\theta)\right)\right\}\right]^2\to 0.$
	
	By \eqref{eq:s1}, for each fixed $\bm\theta$
	\begin{align*}
		\Lambda_n(\bm\theta)=\frac{1}{2}\bm\theta^\top { \bf H}({\bm\beta}^{\dagger})\bm\theta+{\sqrt{n}}{ \bm W_n}^\top{\bm\theta}+o_P(1).
	\end{align*}
	
	Last, we devote to giving the Bahadur representation of $\widetilde{\bm\beta}$.
	Let $\bm\kappa_n=-\sqrt{n}{ \bf H}({\bm\beta}^{\dagger})^{-1}{\bm W}_n$ and $\bm\Theta$ be a convex open subset in $\mathbb{R}^{p+1}$. By Convexity Lemma in \citet{pollard1991asymptotics}, we have
	\begin{align*}
		\Lambda_n(\bm\theta)=\frac{1}{2}\left(\bm\theta-\bm\kappa_n\right)^\top {\bf H}({\bm\beta}^{\dagger})\left(\bm\theta-\bm\kappa_n\right)-\frac{1}{2}\bm\kappa_n^\top {\bf H}({\bm\beta}^{\dagger})\bm\kappa_n + r_n(\bm\theta),
	\end{align*}
	where for each compact set $K$ of $\bm\Theta$, the aforementioned  part is shown for every $\bm\theta\in\bm\Theta$, and then we have $\sup_{{\bm\theta}\in K}\left|r_n(\bm\theta)\right|\to 0$ in probability. Lemma \ref{Asynormalsubestimate} shows that $\bm\kappa_n$ is asymptotically normal which will be proved in the next section, then there exists a compact set $K\in \mathcal{B}_\rho$ with probability close to one, where $\mathcal{B}_\rho$ is a closed ball with center $\bm\kappa_n$ and radius $\rho$. Let $\Delta_n=\sup_{\bm\theta\in\mathcal{B}_\rho}\left|r_n(\bm\theta)\right|$. Then we have
	\begin{align}\label{Deltainside}
		\Delta_n\to 0  ~~\mbox{in probability}.
	\end{align}
	
	Next, we discuss the behavior of $\Lambda_n(\bm\theta)$ outside the closed ball $\mathcal{B}_{\rho}$. Consider $\bm\theta=\bm\kappa_n+\gamma{\bm e}$, with $\gamma>\rho$ and the unit vector ${\bm e}$. A boundary point $\bm\theta^\dagger=\bm\kappa_n+\rho{\bm e}$. Under Assumptions \ref{density}--\ref{derivativeHS} and a similar discussion in Lemma 5 of \citet{koo2008bahadur}, there exists a constant $c_3$ such that $\bm\beta^\top {\bf H}(\bm\beta^\dagger)\bm\beta\geq c_3\|\bm\beta\|^2$. Then, by the convexity of $\Lambda_n(\bm\theta)$ and the definition of $\Delta_n$, we have
	\begin{align*}
		\frac{\rho}{\gamma}\Lambda_n(\bm\theta)+\left(1-\frac{\rho}{\gamma}\right)\Lambda_n(\bm\kappa_n)&\geq\Lambda_n\left(\frac{\rho}{\bm\gamma}\bm\theta+\left(1-\frac{\rho}{\gamma}\right)\bm\kappa_n\right)\\
		&=\Lambda_n(\bm\theta^\dagger)\\
		&\geq\frac{1}{2}\left(\bm\theta-\bm\kappa_n\right)^\top {\bf H}({\bm\beta}^{\dagger})\left(\bm\theta-\bm\kappa_n\right)-\frac{1}{2}\kappa_n^\top {\bf H}({\bm\beta}^{\dagger})\bm\kappa_n-\Delta_n\\
		&\geq \frac{c_3}{2}\rho^2+\Lambda_n(\bm\kappa_n)-2\Delta_n,
	\end{align*}
	which implies that
	\begin{align*}
		\inf_{\|\bm\theta-\bm\kappa_n\|>\rho} \Lambda_n(\bm\theta)\geq \Lambda_n(\bm\kappa_n)+\left(\frac{c_3}{2}\rho^2-2\Delta_n\right).
	\end{align*}
	
	By \eqref{Deltainside}, we can take $\Delta_n$ such that $2\Delta_n<c_3\rho^2/2$
	with probability tending to one.
	Thus $\inf_{\|\bm\theta-\bm\kappa_n\|>\rho} \Lambda_n(\bm\theta)\geq \Lambda_n(\bm\kappa_n).$ This implies the minimum of $\Lambda_n(\bm\theta)$ cannot occur at any $\bm\theta$ with $\|\bm\theta-\bm\kappa_n\|>\rho$. Hence for each $\rho>0$ and let $\widetilde{\bm\theta}_n=\sqrt{n}(\widetilde{\bm\beta}-{\bm\beta}^{\dagger})$,
	we have $\P(\|\widetilde{\bm\theta}_n-\bm\kappa_n\|>\rho) \to 0$. Thus
	\begin{align*}
		\sqrt{n}(\widetilde{\bm\beta}- {\bm\beta}^{\dagger}) = -{\sqrt{n}}{\bf H}({\bm\beta}^{\dagger})^{-1}{ \bm W}_n+o_{P}\left(1\right).
	\end{align*}
	
	The theorem follows the above arguments. \hfill$\Box$
\end{proof}

\section*{Appendix C: Proof of asymptotic normality}

Recall that
\begin{align}
	{\bm M}&=\sum\limits_{i=1}^n { \bm M}_i=\sum\limits_{i=1}^n\frac{1}{nN\pi_i^*}\xi_i^*Y_i^*\widetilde{{\X}}_i^*-\sum\limits_{i=1}^n\left(\frac{1}{nN}\sum\limits_{j=1}^{N}\xi_jY_j\widetilde{{\X}}_j\right),\label{eq:s6}\\
	{\bm Q}&=\frac{1}{N}\sum\limits_{j=1}^N\xi_jY_j\widetilde{{\X}}_j,~~ {\bm T}=\frac{1}{n}\sum\limits_{i=1}^n\frac{1}{N\pi_i^*}\xi_i^*Y_i^*\widetilde{{\X}}_i^*,~~ {\bf B}_N={\bf V}_{T}^{-1/2}{\bf V}_{ M}{\bf V}_{T}^{-1/2},\nonumber
\end{align}
where ${\bf V}_T$ and ${\bf V}_M$ are the variances of $\bm T$ and $\bm M$.
\begin{Lem}\label{MDS}
	$\left\{{\bm M}_i,i=1,\ldots,n\right\}$ {in \eqref{eq:s6}} is a martingale difference sequence relative to the filtration $\left\{\mathcal{F}_{N,i},i=1,\ldots,n\right\}$.
\end{Lem}

\begin{proof}
	The $\mathcal{F}_{n,i}$-measurability follows from the definition of $\bm M_i$ and the definition of the  filtration $\left\{\mathcal{F}_{N,i},i=1,\ldots,n\right\}$.  Moreover, we have
	\begin{align*}
		{\E}\left\{{\bm M}_i\mid \mathcal{F}_{N,i-1}\right\} &= {\E}_{Y\mid \X}\left\{\frac{1}{nN\pi_i^*}\xi_i^*Y_i^*\widetilde{{\X}}_i^*\right\}-\frac{1}{nN}\sum\limits_{j=1}^N\xi_jY_j\widetilde{{\X}}_j\\
		&=\frac{1}{nN}\sum\limits_{i=1}^N\xi_iY_i\widetilde{{\X}}_i-\frac{1}{nN}\sum\limits_{j=1}^N\xi_jY_j\widetilde{{\X}}_j\\
		&=0,
	\end{align*}
	where
	${\E}_{Y\mid \X}$ is the expectation with respect to sampling randomness or the conditional expectation of
	$Y$ given ${\X}_1^N$ with  ${{\X}}_1^N=\left({{\X}}_1,\ldots,{{\X}}_N\right)$.
	Then $\left\{{\bm M}_i,i=1,\ldots,n\right\}$ is a martingale difference sequence. \hfill$\Box$
\end{proof}

\begin{Lem}\label{Ling-cond3}
	Suppose Assumptions \ref{density} and \ref{approxexpansion} hold. Let ${\bf V}_T$ and ${\bf V}_Q$ denote the variances of $\bm T$ and $\bm Q$. For any $\bm t\in\mathbb{R}^{p+1}$, we have
	\begin{align*}
		\left|{\E}\left\{\exp\left(i{\bm t}^{\top}{\bf V}_{ T}^{-1/2}{\bm  Q}\right)\right\}-{\E}\left\{\exp\left(i {\bm t}^{\top}{\bf V}_{T}^{-1/2}{\bf V}_{Q}^{1/2}{ \bm A}_0\right)\right\}\right|\to 0,
	\end{align*}
	as $N\to \infty$, where  ${\bm A}_0\sim \mathcal{N}({\bf 0},{\bf I}_{p+1})$.
\end{Lem}

\begin{proof}
	Note $\bm Q$ is a sum of i.i.d mean zero random vectors, $\xi_jY_j\widetilde{{\X}}_j$. The Linderberg-Feller conditions are satisfied by Assumption \ref{density} and Assumption \ref{approxexpansion}, then we have
	\begin{align}\label{Asynormalfullestimate}
		{\bf V}_{Q}^{-1/2}{\bm Q} {\rightarrow} \mathcal{N}\left({\bf 0},{\bf I}_{p+1}\right).
	\end{align}
	
	Furthermore, for any $\bm\varsigma\in\mathbb{R}^{p+1}$ and as $N\to \infty$
	\begin{align*}
		\left|{\E}\left\{\exp\left(i\bm\varsigma^{\top}{\bf V}_{Q}^{-1/2}{ \bm Q}\right)\right\}-{\E}\left\{\exp\left(i\bm\varsigma^{\top}{ \bf A}_0\right)\right\}\right|\to 0.
	\end{align*}
	
	Let $\bm\varsigma={\bf V}_{Q}^{1/2}{\bf V}_{ T}^{-1/2}{\bm t}^{\top}$. For any fixed ${\bm t}$, we need to verify the following condition to prove this lemma
	\begin{align*}
		\sup_{N}\|\bm\varsigma\|<\infty.
	\end{align*}
	
	We note that $\|\bm\varsigma\|\leq \sigma_{\max}\left({\bf V}_{ Q}^{1/2}{\bf V}_{ T}^{-1/2}\right)\cdot\|{\bm t}\|$, where $\sigma_{\max}(\cdot)$ denotes the maximum eigenvalue of the corresponding matrix. Hence it is enough to show $\sigma_{\max}({\bf V}_{ Q}^{1/2}{\bf V}_{ T}^{-1/2})\leq 1$. Since the covariance matrix ${\bf V}_{Q}$ and ${\bf V}_{T}$ are positive-defined, the following equation holds
	\begin{align*}
		{\bf V}_{Q}^{1/2}{\bf V}_{T}^{-1/2}={\bf V}_{ T}^{1/4}\left({\bf V}_{T}^{-1/4}{\bf V}_{ Q}^{1/2}{\bf V}_{T}^{-1/4}\right){\bf V}_{T}^{-1/4},
	\end{align*}
	thus ${\bf V}_{Q}^{1/2}{\bf V}_{T}^{-1/2}$ is similar to ${\bf V}_{T}^{-1/4}{\bf V}_{Q}^{1/2}{\bf V}_{T}^{-1/4}$. It only needs to show $\sigma_{\max}({\bf V}_{T}^{-1/4}{\bf V}_{ Q}^{1/2}{\bf V}_{T}^{-1/4})\leq 1$, which is equal to show
	\begin{align*}
		{\bf I}_{p+1}-{\bf V}_{T}^{-1/4}{\bf V}_{ Q}^{1/2}{\bf V}_{ T}^{-1/4}={\bf V}_{T}^{-1/4}\left({\bf V}_{ T}^{1/2}-{\bf V}_{ Q}^{1/2}\right){\bf V}_{T}^{-1/4}>0,
	\end{align*}
	that is equivalent to show ${\bf V}_{T}^{1/2}-{\bf V}_{Q}^{1/2}$ is positive-defined.
	
	Recall that ${\bm M}= {\bm T}-{\bm Q}$ and by Lemma \ref{martingaleCLT}, we have ${\bf V}_{ T}-{\bf V}_{Q}={\bf V}_{M}> 0$. Then by the L\"{o}wner-Heinz theorem in \citet{zhan2004matrix}, we get ${\bf V}_{T}^{1/2}-{\bf V}_{Q}^{1/2}>0$ which completes the proof of this lemma. \hfill$\Box$
\end{proof}

\begin{Lem}\label{Asynormalsubestimate}
	Suppose Assumptions \ref{density} and \ref{approxexpansion} hold. Then we have
	\begin{align*}
		{\bf V}_{T}^{-1/2}{\bm T}{\rightarrow} \mathcal{N}({\bf 0},{\bf I}_{p+1}).
	\end{align*}
\end{Lem}

\begin{proof}
	Recall the conditions in Lemma \ref{martingaleCLT} with
	\[
	\bm\eta_{ki}=\bm\eta_{Ni}, {\bm Z}_k ={\bf V}_{T}^{-1/2}{\bm Q}, {\bf B}_k = {\bf B}_N, \bm L_0 \sim\mathcal{N}({\bf 0},{\bf I}_{p+1}).
	\]
	
	By Lemma \ref{MDS}, $\left\{M_i,i=1,\ldots,n\right\}$ is a martingale difference sequence, then the first two conditions in Lemma \ref{MDS} are easily satisfied by Assumption \ref{density}. It suffices to show the third condition in Lemma \ref{martingaleCLT} holds.
	
	By \eqref{Asynormalfullestimate} in Lemma \ref{Ling-cond3}, we have ${\bf V}_{Q}^{-1/2}{\bm Q} {\rightarrow} \mathcal{N}\left({\bf 0},{\bf I}_{p+1}\right)$. Next, we devote ourselves to verifying the third condition in Lemma \ref{martingaleCLT}. Let ${\bf V}_M$ be the variance of $\bm M$. For any ${\bm t}\in\mathbb{R}^{p+1}$, we have the following characteristic function
	\begin{align*}
		&{\E}\left\{\exp\left(i{\bm t}^{\top}{\bf V}_{T}^{-1/2}{ \bm Q}\right)\right\}\cdot\exp\left(-\frac{1}{2}{\bm t}^{\top}{\bf V}_{ T}^{-1/2}{\bf V}_{M}{\bf V}_{T}^{-1/2}{\bm t}\right)\\
		=&{\left\{\exp\left(i{\bm t}^{\top}{\bf V}_{ T}^{-1/2}{\bf V}_{ Q}{\bf V}_{T}^{-1/2}{\bm t}\right)+o(1)\right\} }\cdot\exp\left(-\frac{1}{2}{\bm t}^{\top}{\bf V}_{ T}^{-1/2}{\bf V}_{M}{\bf V}_{T}^{-1/2}{\bm t}\right)\\
		=&\left\{\exp\left(i{\bm t}^{\top}{\bf V}_{ T}^{-1/2}{\bf V}_{Q}{\bf V}_{ T}^{-1/2}{\bm t}\right)\right\}\cdot\exp\left(-\frac{1}{2}{\bm t}^{\top}{\bf V}_{T}^{-1/2}{\bf V}_{M}{\bf V}_{T}^{-1/2}{\bm t}\right)+o(1)\\
		=&\exp\left(-\frac{1}{2}{\bm t}^{\top}{\bm t}\right)+o(1),
	\end{align*}
	where the first equality holds by Lemma \ref{Ling-cond3}. And the third condition in Lemma \ref{martingaleCLT} is satisfied. Then by Lemma \ref{martingaleCLT} and \eqref{Asynormalfullestimate} we have
	\begin{align*}
		{\bf V}_{T}^{-1/2}{\bm Q}+{\bf V}_{T}^{-1/2}{ \bm M}={\bf V}_{ T}^{-1/2}{\bm T}{\rightarrow} \mathcal{N}\left({\bf 0},{\bf I}_{p+1}\right).
	\end{align*} \hfill$\Box$
\end{proof}

\bigskip
\noindent
\textbf{Proof of Theorem~\ref{martingalenormality}.} By Theorem \ref{Bahadur-type representation} and Lemma \ref{Asynormalsubestimate}, we have
\begin{align*}
	\sqrt{n}(\widetilde{\bm\beta}-{\bm\beta}^{\dagger}) = -\sqrt{n}{\bf H}(\bm\beta^{\dagger})^{-1}{\bm T}+o_p(1).
\end{align*}

It follows that
\begin{align*}
	{\bf V}_{T}^{-1/2}{ \bf H}(\bm\beta^{\dagger})(\widetilde{\bm\beta}-\bm\beta^{\dagger})+o_p(1)=-{\bf V}_{ T}^{-1/2}{\bm T}.
\end{align*}

By Lemma \ref{Asynormalsubestimate}, we have
\begin{align*}
	{\bf V}^{-1/2}(\widetilde{\bm\beta}-{\bm\beta}^{\dagger}) {\rightarrow} \mathcal{N}({\bf 0},{\bf I}_{p+1}),
\end{align*}
where ${\bf V}={\bf H}({\bm\beta}^{\dagger})^{-1}{\bf V}_{ T}{\bf H}({\bm\beta}^{\dagger})^{-1}$. $\hfill\square$

\section*{Appendix D: Proof of Theorem \ref{Samplingprob}}

\noindent
\textbf{Proof of Theorem~\ref{Samplingprob}.} Recall that ${{\X}}_1^N=\left({{\X}}_1,\ldots,{{\X}}_N\right)$ and ${{Y}}_1^N=\left({{Y}}_1,\ldots,{{Y}}_N\right)$, then  $\mathcal{D}_N = \left\{{{\X}}_1^N,{{Y}}_1^N\right\}$. {Let $\mathrm{var}(Y\mid \X)$ be the conditional variance of $Y$ given $\X$.} First we calculate  $\mathrm{var}({\bm T}\mid {{\X}}_1^N )$.
We have
\begin{align*}
	\mathrm{var}({\bm T}\mid {{\X}}_1^N) = {\E}_{Y\mid \X}\left\{\mathrm{var}({\bm T}\mid \mathcal{D}_N)\right\}+ \mathrm{var}_{Y\mid \X}\left\{{\E}({\bm T}\mid \mathcal{D}_N)\right\}.
\end{align*}

Some algebra yields
\begin{align}\label{varyE}
	\begin{split}
		\mathrm{var}_{Y\mid \X}\left\{{\E}({\bm T}\mid \mathcal{D}_N)\right\}
		&=\mathrm{var}_{Y\mid \X}\left(\frac{1}{N}\sum\limits_{j=1}^N\xi_jY_j\widetilde{{\X}}_j\right)\\
		&=\frac{1}{N^2}\sum\limits_{j=1}^N{\E}_{Y\mid \X}\left(\xi_j^2Y_j^2\widetilde{{\X}}_j\widetilde{{\X}}_j^{\top}\right)-\frac{1}{N^2}\sum\limits_{j=1}^N\left\{{\E}_{Y\mid \X}(\xi_jY_j{\widetilde{\X}}_j)\right\}^2\\
		&=\frac{1}{N^2}\sum\limits_{j=1}^N{\E}_{Y\mid \X}\left(\xi_j\widetilde{{\X}}_j\widetilde{{\X}}_j^{\top}\right)-\frac{1}{N^2}\sum\limits_{j=1}^N\left\{{\E}_{Y\mid \X}(\xi_jY_j{\widetilde{\X}}_j)\right\}^2,
	\end{split}
\end{align}
where the third equality holds by the fact that $\xi_j^2=\xi_j$ and $Y_j^2=1$. Next
\begin{align}\label{Eyvar}
	\begin{split}
		{\E}_{Y\mid \X}\left\{\mathrm{var}(\bm T\mid \mathcal{D}_N)\right\}
		&=\frac{1}{nN^2}\sum\limits_{j=1}^N{\E}_{Y\mid \X}\left\{\pi_j\left(\frac{1}{\pi_j^2}\xi_j^2Y_j^2\widetilde{{{\X}}}_j\widetilde{{{\X}}}_j^\top\right)\right\}-\frac{1}{nN}\sum\limits_{j=1}^N\left\{{\E}_{Y\mid \X}(\xi_jY_j\widetilde{{{\X}}}_j)\right\}^2\\
		&=\frac{1}{nN^2}\sum\limits_{j=1}^N{\E}_{Y\mid \X}\left\{\frac{1}{\pi_j}\xi_j\widetilde{{\X}}_j\widetilde{{\X}}_j^\top\right\}-\frac{1}{nN}\sum\limits_{j=1}^N\left\{{\E}_{Y\mid \X}(\xi_jY_j\widetilde{{{\X}}}_j)\right\}^2.
	\end{split}
\end{align}

In view of (\ref{varyE}) and (\ref{Eyvar}), we  get
\begin{align*}
	\mathrm{var}({\bm T}\mid {{\X}}_1^N) =&\frac{1}{nN^2}\sum\limits_{j=1}^N{\E}_{Y\mid \X}\left(\frac{1}{\pi_j}\xi_j\widetilde{{\X}}_j\widetilde{{\X}}_j^\top\right)+ \frac{1}{N^2}\sum\limits_{j=1}^N{\E}_{Y\mid \X}\left(\xi_j\widetilde{{\X}}_j\widetilde{{\X}}_j^\top\right)\\
	&-\frac{1}{N}\sum\limits_{j=1}^N\left\{{\E}_{Y\mid \X}\left(\xi_jY_j\widetilde{{{\X}}}_j\right)\right\}^2\left(\frac{1}{N}+\frac{1}{n}\right).
\end{align*}

Next we calculate ${\bf V}_T$
through
\begin{align*}
	{\bf V}_{T}= {\E}\left\{\mathrm{var}( \bm T\mid{{\X}}_1^N)\right\} + \mathrm{var}\left\{{\E}( \bm T\mid {{\X}}_1^N)\right\}.
\end{align*}

A simple calculation shows that
\begin{align*}
	&{\E}(T\mid {{\X}}_1^N)={\E}\left\{{\E}( \bm T\mid {{\X}}_1^N,Y_1^N) \right\}=\frac{1}{N}\sum\limits_{j=1}^N{\E}_{Y\mid \X}\left(\xi_jY_j\widetilde{{\X}}_j\right),\\
	&\mathrm{var}\left\{{\E}(T\mid {{\X}}_1^N)\right\}=\frac{1}{N^2}\sum\limits_{j=1}^N{\E}_{Y\mid \X}\left(\xi_j\widetilde{{\X}}_j\widetilde{{\X}}_j^\top\right)-\frac{1}{N^2}\sum\limits_{j=1}^N\left\{{\E}_{Y\mid \X}\left(\xi_jY_j\widetilde{{\X}}_j\right)\right\}^2.
\end{align*}

Therefore, we have
\begin{align*}
	{\bf V}_{T} =\frac{1}{nN^2}\sum\limits_{j=1}^N{\E}_{Y\mid \X}\left(\frac{1}{\pi_j}\xi_j\widetilde{{\X}}_j\widetilde{{\X}}_j^\top\right)+{\bf C},
\end{align*}
where ${ \bf C}={2}{N^{-2}}\sum\nolimits_{j=1}^N\E_{Y\mid \X}\left(\xi_j\widetilde{{\X}}_j\widetilde{{\X}}_j^\top\right)-{N^{-1}}\sum\nolimits_{j=1}^N\left\{{\E}_{Y\mid \X}\left(\xi_jY_j\widetilde{{{\X}}}_j\right)\right\}^2\left({2}{N^{-1}}+{n^{-1}}\right)$ is a constant matrix that
does not depend on $\bm\pi$.

Let $\mathrm{tr}({\bf A})$ denotes the trace of matrix ${\bf A}$. We minimize $\mathrm{tr}({\bf V}_{T})$ to obtain the A-optimality subsampling probability

\begin{align*}
	\mathrm{tr}\left({\bf V}_{ T}\right)&=\frac{1}{nN^2}\sum\limits_{j=1}^N\mathrm{tr}\left\{\E_{Y\mid \X}\left(\frac{1}{\pi_j}\xi_j{ \bf H}(\bm\beta^{\dagger})^{-1}\widetilde{{\X}}_j\widetilde{{\X}}_j^\top {\bf H}(\bm\beta^{\dagger})^{-1}\right)\right\}+\mathrm{tr}({\bf C})\\
	&=\frac{1}{nN^2}\E_{Y\mid \X}\left\{{\sum\limits_{j=1}^N\pi_j}\sum\limits_{j=1}^N\left(\frac{1}{\pi_j}\xi_j\|{ \bf H}({\bm\beta}^{\dagger})^{-1}\widetilde{{\X}}_j\|^2\right)\right\}+\mathrm{tr}({\ \bf C})\\
	&\geq \frac{1}{nN^2}\left\{\sum\limits_{j=1}^N\P\left(Y_jf(\ {{\X}}_j,{\bm\beta}^{\dagger})\leq 1\right)\|{ \bf H}({\bm\beta}^{\dagger})^{-1}\widetilde{{\X}}_j\|\right\}^2+\mathrm{tr}({ \bf C}),
\end{align*}
where the last inequality follows from the Cauchy-Schwarz inequality, and the equality holds if and only if
\begin{align*}
	\pi_j^{\text{A}}\varpropto \I\left(Y_jf(\ {{\X}}_j,{\bm\beta}^{\dagger})\leq 1\right)\|{ \bf H}({\bm\beta}^{\dagger})^{-1}\widetilde{{\X}}_j\|.
\end{align*}

Note that  ${ \bf H}(\bm\beta^{\dagger})^{-1}\mathrm{var}(\bm T\mid {{\X}}_1^N){\bf H}(\bm\beta^{\dagger})^{-1}$ depends on subsampling probability $\pi$
only through $\mathrm{var}( \bm T\mid {{\X}}_1^N)$. Hence,
by the similar argument for minimizing $\mathrm{tr}\left\{\mathrm{var}(\bm T\mid {{\X}}_1^N)\right\}$,  we get the L-optimality subsampling probability
\begin{align*}
	\pi_j^{\text{L}}\varpropto \I\left(Y_jf(\ {{\X}}_j,{\bm\beta}^{\dagger})\leq 1\right)\|\widetilde{{\X}}_j\|.
\end{align*} $\hfill\square$

\section*{Appendix E: Additional simulation results}

\begin{figure}[htbp]
	\centering
	\setlength{\abovecaptionskip}{0cm}
	\setlength{\belowcaptionskip}{-0.cm}
	\includegraphics[width=1\textwidth]{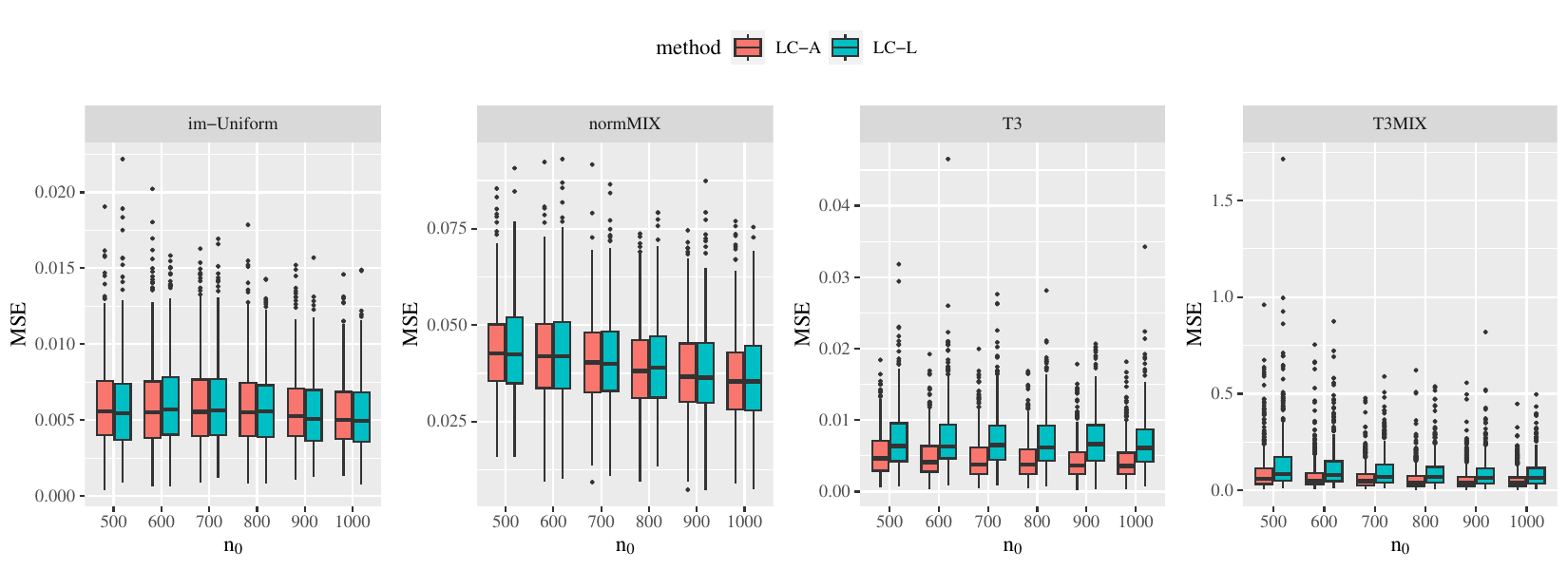}\par
	\vspace{-0.2cm}
	\caption{Comparison of MSE for approximating the full sample SVM estimator $\widehat{\bm\beta}$ with different pilot subsample sizes given 
		$n=1000$ under Scenarios I--IV.}\label{pilot-estimation}
\end{figure}

To assess the impact of the pilot study in our proposed algorithm, we conduct the following boxplot by 500 replications on the four scenarios presented in Section \ref{Simulation}. Figure \ref{pilot-estimation} reveals that the MSE is not sensitive to the pilot subsample size $n_0$. 
As $n_0$ increases, the boxplot shows a slight decrease in MSE, suggesting that a smaller pilot subsample size can reduce computational costs without significantly compromising accuracy.

\begin{figure}[htbp]
	\centering
	\setlength{\abovecaptionskip}{0cm}
	\setlength{\belowcaptionskip}{-0.cm}
	\includegraphics[angle=0,scale=0.65]{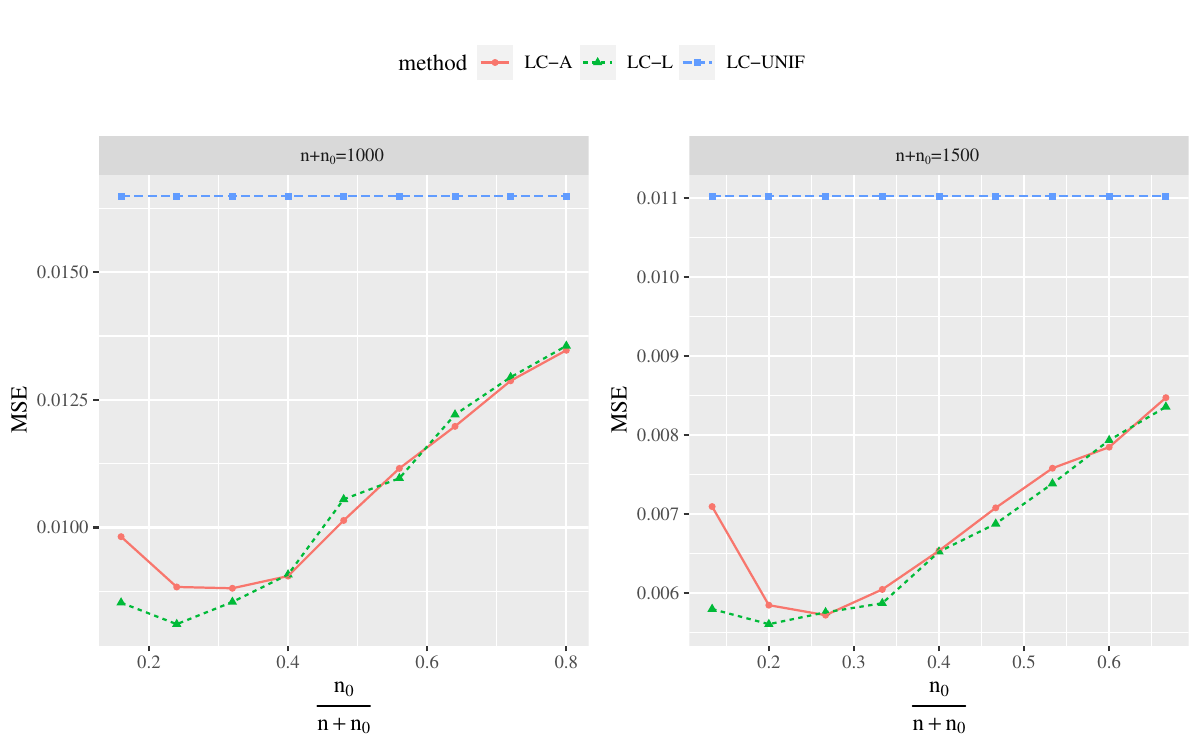}\par
	\vspace{-0.2cm}
	\caption{Comparison of mean squared errors (MSEs) for approximating the full sample SVM estimator $\widehat{\bm\beta}$ with different subsample size allocations under Scenario I.}\label{unif-fix}
\end{figure}

Moreover, we fix the total subsample size of $n+n_0$ and vary the proportions of $n$ and $n_0$. It provides practical guidelines on allocating subsamples in two steps. We evaluate both $\widehat{\bm\pi}^{\text{A}}$ and $\widehat{\bm\pi}^{\text{L}}$ and the results are presented in Figure \ref{unif-fix} under Scenario I. It illustrates that the MSEs increase when $n_0$ is either too small or too large. This is because that if $n_0$ is too small, the pilot estimate is not accurate, and thus the optimal subsampling probabilities may not be well approximated; on the other hand, if $n_0$ is too large, there is not enough sampling budget to select informative subsample in subsequent steps. Figure \ref{unif-fix} shows that our methods perform well when the ratio $n_0/(n+n_0)$ is around $(0.2,0.4)$. Therefore, we use $n_0=500$ in our simulation studies with $N=10^5$. 

Bandwidth selection is a critical issue in nonparametric estimation. In Table \ref{Tab:BandwidthSelector}, we compare the MSE and accuracy of LC-A with three bandwidth selectors: Silverman's rule of thumb   \citep[ROT,][]{silverman1986density},  Sheather and Jones method,  \citep[SJ,][]{sheather1991reliable}, and biased cross-validation, \citep[BCV,][]{scott1987biased}. Clearly, The results demonstrate that the choice of bandwidth selector has a negligible impact on the empirical MSE and accuracy. To this end, we employ the commonly-used bandwidth selector, Silverman's rule of thumb  \citep{silverman1986density}, in our numerical analysis.

\begin{table}[htbp]
	\tabcolsep 6pt
	\caption{Comparison of MSE (1$0^{-2}$) and prediction accuracy (\%) for LC-A against different bandwidth selectors under Scenarios I--II when $n=1000$.}
	\vspace{0.3cm}
	\label{Tab:BandwidthSelector}
	\centering
	{\centering
		\scalebox{1}{
			\begin{tabular}{lcrrrrrrrrrrrrrrrr}
			\hline
				& & \multicolumn{2}{c}{\textbf{ROT}} &&\multicolumn{2}{c}{\textbf{SJ}}
				&&\multicolumn{2}{c}{\textbf{BCV}}\\
				\cline{3-4} \cline{6-7} \cline{9-10}  
				\text{Scenario}&$n_0$& MSE&Accuracy&&MSE&Accoracy&&MSE&Accuracy\\
				\hline
				&300&0.68&95.54&&0.92&94.52&& 0.65&94.56\\
				{\bf im-Uniform}&400&0.64&94.53&&0.85&94.52&&0.61&94.54\\
				&500&0.60&94.53&&0.75&94.52&&0.60&94.53\\
				\hline
				&300& 4.84&97.52&&4.89&97.52&&4.87&97.52\\
				{\bf normMIX}&400& 4.49&97.53&&4.63&97.53&&4.56&97.53\\
				&500&4.33&97.54&&4.43&97.54&&4.35&97.54\\
				\hline
	\end{tabular}}}\\
\end{table}


	\end{CJK*}	
\end{document}